\documentclass[a4paper, notitlepage, 11pt]{article}
\usepackage[utf8]{inputenc}
\usepackage[english]{babel}
\usepackage[T1]{fontenc}
\usepackage{amsmath,amsthm,amssymb,amscd}
\usepackage{color}
\usepackage{soul}
\usepackage{url}
\usepackage[normalem]{ulem}
\usepackage{enumitem}
\usepackage[left=2.5cm,right=2.5cm,bottom=2.5cm,top=2.5cm]{geometry}
\usepackage{tikz}
\usepackage{appendix}
\usepackage{csquotes}
\usepackage{float}
\usepackage{fancyhdr}
\usepackage{mathtools}
\usepackage{stmaryrd}
\usepackage{thm-restate}

\usepackage{diagbox}

\numberwithin{equation}{section}

\newcommand{\blue}[1]{{\color{blue} #1}}

\newcommand{\bra}[1]{\langle #1|}
\newcommand{\cket}[1]{|#1\rangle}
\newcommand{\bracket}[2]{\langle #1|#2\rangle}
\newcommand{\mcl}[1]{\mathcal{#1}}
\newcommand{\norme}[2]{\left|\left|#1\right|\right|_{#2}}

\newcommand{\Acal}{{\mathcal A}}
\newcommand{\Bcal}{{\mathcal B}}
\newcommand{\Dcal}{{\mathcal D}}

\newcommand{\EcalKDC}{{\mathcal E}_{\mathrm{KD+}}}
\newcommand{\EcalKDCpu}{{\mathcal E}_{\mathrm{KD+}}^{\mathrm{pure}}}
\newcommand{\EcalKDCext}{{\mathcal E}_{\mathrm{KD+}}^{\mathrm{ext}}}

\newcommand{\spanRAB}{{\mathrm{span}_{\R}(\Acal\cup\Bcal)}}

\newcommand{\Hcal}{{\mathcal H}}

\newcommand{\Ker}{{\mathrm{Ker}}}

\newcommand{\N}{\mathbb{N}}
\newcommand{\Z}{\mathbb{Z}}
\newcommand{\R}{\mathbb{R}}

\newcommand{\C}{\mathbb{C}}

\newcommand{\bbone}{\mathbb{I}}

\newcommand{\nab}{n_{\Acal,\Bcal}}

\newcommand{\na}{n_{\Acal}}
\newcommand{\nb}{n_{\Bcal}}
\newcommand{\mab}{m_{\Acal,\Bcal}}
\newcommand{\Mab}{M_{\Acal, \Bcal}}

\newcommand{\Tr}{\mathrm{Tr}\,}
\newcommand{\conv}[1]{\mathrm{conv}\left(#1\right)}

\newcommand{\VR}{V_{\mathrm{KDr}}}
\newcommand{\VRp}{V_{\mathrm{KD+}}}

\renewcommand{\Im}[1]{\mathrm{Im}#1}
\newcommand{\Ran}[1]{\mathrm{Ran}\left(#1\right)}
\renewcommand{\H}{\mathcal{H}}
\newcommand{\SAO}{\mathcal{S}_{d}}
\newcommand{\SAOr}{\mathcal{S}_{d,\mathrm r}}
\newcommand{\SAOpone}{\mathcal{S}_{d,+,1}}
\newcommand{\SAOp}{\mathcal{S}_{d,+}}
\newcommand{\MatC}[1]{\mcl{M}_{#1}(\C)}
\newcommand{\MatR}[1]{\mcl{M}_{#1}(\R)}
\newcommand{\IntEnt}[2]{\llbracket #1 , #2 \rrbracket}
\newcommand{\convAB}{\conv{\Acal \cup \Bcal}}
\renewcommand{\epsilon}{\varepsilon}

	\renewcommand{\thesection}{\arabic{section}}
	\renewcommand{\thesubsection}{\arabic{section}.\arabic{subsection}}

\newtheorem{Theorem}{Theorem}[section]
\newtheorem*{thmIntro}{Theorem \ref{thm:Graal}}
\newtheorem*{propIntro}{Proposition \ref{prop:pure_char}}

\newtheorem{Cor}[Theorem]{Corollary}

\newtheorem{Lemma}[Theorem]{Lemma}
\newtheorem{Prop}[Theorem]{Proposition}

\title{Characterizing the geometry of the Kirkwood-Dirac positive states}
\author{C. Langrenez$^1$\thanks{christopher.langrenez@univ-lille.fr}, D.R.M. Arvidsson-Shukur$^2$\thanks{drma2@cam.ac.uk}, S. De Bi\`evre$^1$\thanks{stephan.de-bievre@univ-lille.fr}\\
$\,^1$Univ. Lille, CNRS, Inria, UMR 8524, Laboratoire Paul Painlev\'e, F-59000 Lille, France\\
$\,^2$ Hitachi Cambridge Lab., J.J. Thomson Avenue, Cambridge CB3 0HE, UK 
}
\begin{document}

\maketitle

\begin{abstract}

The Kirkwood-Dirac (KD) quasiprobability distribution can describe any quantum state with respect to the eigenbases of two  observables $A$ and $B$. KD distributions behave similarly to classical joint probability distributions but can assume negative and nonreal values. In recent years, KD distributions have proven instrumental in mapping out nonclassical phenomena and quantum advantages. These quantum features have been connected to nonpositive entries of KD distributions. Consequently, it is important to understand the geometry of the KD-positive and -nonpositive states. Until now, there has been no thorough analysis of the KD positivity of mixed states. Here, we characterize how the full convex set of states with positive KD distributions depends on the eigenbases of $A$ and $B$. In particular, we identify three regimes where convex combinations of the eigenprojectors of $A$ and $B$ constitute the only KD-positive states: $(i)$ any system in dimension $2$; $(ii)$ an open and dense set of bases in dimension $3$; and $(iii)$ the discrete-Fourier-transform bases in prime dimension. Finally, we investigate if there can exist mixed KD-positive states that cannot be written as convex combinations of pure KD-positive states. We show that for some choices of observables $A$ and $B$ this phenomenon does indeed occur. We explicitly construct such states for a spin-$1$ system.

\end{abstract}
\section{Introduction}

In classical mechanics,  a joint probability distribution $\mathcal{P}(\bf{x}, \bf{p})$ can describe a system with respect to two observables, such as position $\bf{x}$ and momentum $\bf{p}$. In quantum mechanics, however, observables generally do not commute and probabilistic descriptions of states with respect to more than one observable are often not available \cite{wigner1932,hudson1974, cohen1966, srinivaswolf1975, hartle2004, allahverdyan2015}.  
Nevertheless, one can describe a quantum state with respect to two joint observables via a \textit{quasiprobability} distribution.   Quasiprobability distributions obey all but one of Kolmogorov's axioms for probability distributions \cite{brookeskolmogorov1951}: their entries sum to unity;  their marginals correspond to the probability distributions given by the Born Rule; but individual quasiprobabilities can take negative or  nonreal values.  

The quasiprobability formalism provides a useful alternative to other descriptions of quantum states.  The most famous quasiprobability distribution is the Wigner function. It deals with continuous-variable systems with clear analogues of position and momentum. Most notably, the Wigner function has played a pivotal role in the analyses of quantum states of light~\cite{cagl69a, cagl69b, leonhardt2010, serafini2017}. The Wigner function, and other quasiprobability distributions {\cite{Husimi40, Sudarshan63, Glauber63, Terletsky37, Margenau61, Johansen04, Johansen04-2}}, allow techniques from statistics and probability theory to be applied to quantum mechanics. 

Most  modern quantum-information research is phrased in terms of finite-dimensional systems---often systems of qubits. Moreover, the observables of interest are, unlike position and momentum, not necessarily conjugate. The Wigner function is ill-suited for such systems and observables. Instead,  recent years have seen a different  quasiprobability distribution come to the foreground: the Kirkwood-Dirac (KD) distribution {\cite{kirkwood1933, dirac1945, YuSwDr18, arvidsson-shukuretal2021, debievre2021, debievre2023a}}. 

The KD distribution has proven itself a tremendously versatile tool in studying and developing  quantum-information processing.  In its standard form, the KD distribution describes a quantum state $\rho$ with respect to two orthonormal bases  $ \left(\cket{a_i}\right)_{i\in\IntEnt{1}{d}}$ and $\left(\cket{b_j}\right)_{j\in\IntEnt{1}{d}}$ in a complex Hilbert space $\H$ of dimension $d$. The KD distribution  reads
\begin{equation}\label{eq:KDGV1}
\forall (i,j)\in \IntEnt{1}{d}^2, \; Q_{ij}(\rho) = \bracket{b_j}{a_i}\bra{a_i}\rho\cket{b_j}.
\end{equation}

By associating the two bases with the eigenstates of observables of interest, the KD distribution can be tuned towards a specific problem.  So far, the KD distribution has been used  to study, describe or develop: direct state tomography {\cite{Johansen04, Lundeen11, lundeenbamber2012, bamberlundeen2014, thekkadathetal2016}}; quantum metrology {\cite{arvidsson-shukuretal2020, Jenne22, lupu-gladsteinetal2021}}; quantum chaos {\cite{Yunger17, YuSwDr18, Yunger18-2, GoHaDr19, Razieh19}}; weak measurements {\cite{Aharonov88, Duck89, Dressel14,pusey2014, Dressel12, YuSwDr18,Kunjwal19, Monroe20, Wagner23}}; quantum thermodynamics {\cite{Yunger17, lostaglioetal2022, Levy20, Lostaglio20, Hernandez23, Upadhyaya23}}; quantum scrambling {\cite{YuSwDr18, Yunger18-2}}; Leggett-Garg inequalities {\cite{Suzuki12, Suzuki16, Halliwell16}}; generalised contextuality {\cite{pusey2014, Levy20, Lostaglio20}}; consistent-histories interpretations of quantum mechanics {\cite{Griffiths84}}; measurement disturbance {\cite{hofmann2011, Dressel12, dressel2015, debievre2021, debievre2023a, Fiorentino22, Gao23}}{; and coherence \cite{Budiyono23}}. The list can be made longer, but the point is clear: the Kirkwood-Dirac distribution  currently experiences great prosperity and growing interest.  

Below, a state will be said to be KD positive when its KD distribution only takes on positive or zero values. Such states have been called KD classical elsewhere~\cite{arvidsson-shukuretal2021, debievre2021, debievre2023a}. We prefer to avoid this terminology since the terms ``classical'' and ``nonclassical''  lack unique definitions. The capacity of quasiprobability distributions to describe quantum phenomena hinges on their ability to assume negative or nonreal values.  An always-positive (probability) distribution cannot describe all of quantum mechanics. As concerns the KD distribution, nonpositive quasiprobabilities have been linked to {various forms of} quantum advantages in, {for example, weak measurements \cite{pusey2014, Kunjwal19},  quantum metrology \cite{arvidsson-shukuretal2020, Jenne22} and quantum thermodynamics \cite{YuSwDr18,Levy20,Lostaglio20}.} Therefore, it is important to understand: {When does a KD  distribution  assume only positive or zero values?} While this question has been addressed for pure states~\cite{arvidsson-shukuretal2021, debievre2021, debievre2023a}, a general study of the mixed KD-positive states is lacking. {In this work, we provide such a study.}

To analyse how the KD distribution underlies nonclassical phenomena, one {must first} understand the geometric structure of the convex set $\EcalKDC$ of KD-positive states.  We know, by the Krein-Milman theorem \cite{hiriart-urrutylemarechal2001}, that this set is the convex hull of the set $\EcalKDCext$ of its extreme points:
$$
\EcalKDC=\conv{\EcalKDCext}.
$$
{It is, therefore, desirable}  to have a full description and a convenient characterization of $\EcalKDCext$. The set $\EcalKDCext$  always contains all the  basis states $ \left(\cket{a_i}\right)_{i\in\IntEnt{1}{d}}$ and $\left(\cket{b_j}\right)_{j\in\IntEnt{1}{d}}$. {Additionally, $\EcalKDCext$}  may contain other pure and also mixed states. Experience with similar {analyses} for the Wigner function, where  {the mixed-state characterization of Wigner positive states is} not  fully solved, indicates that it might be {difficult to obtain a full characterization of  $\EcalKDCext$ for } general KD distributions.

 Our results {about the convex set $\EcalKDC$ of KD-positive states} can be summed up as follows. We first identify for what choices of the bases $ \left(\cket{a_i}\right)_{i\in\IntEnt{1}{d}}$ and $\left(\cket{b_j}\right)_{j\in\IntEnt{1}{d}}$ the only  KD-positive states are those that are convex mixtures of the basis states. The following theorem provides a precise statement of these results.  We introduce 
\begin{eqnarray*}
    \Acal=\{ |a_i\rangle\langle a_i|\mid i\in\IntEnt{1}{d}\},\quad
    \Bcal=\{ |b_j\rangle\langle b_j|\mid j\in\IntEnt{1}{d}\} ,
\end{eqnarray*}
which are the families of rank-one projectors associated to the two bases. Also, we write $U_{ij}=\langle a_i|b_j\rangle$ for the transition matrix between the two bases and introduce
$$
\mab=\min_{i,j}|\langle a_i|b_j\rangle|.
$$
\begin{Theorem}
\label{thm:Graal}
The equality
\begin{equation}\label{eq:Graalagain}
\EcalKDC=\convAB,
\end{equation}
holds under any single one of the following hypotheses:
\begin{enumerate}[label=(\roman*),wide, labelindent=0pt]
    \item If $d=2$ (for qubits) and $\mab>0$;
    \item If $d=3$, for all $U$ in a open dense set of probability $1$;
    \item If $d$ is prime and $U$ is the discrete Fourier transform (DFT) matrix;
    \item If $U$ is sufficiently close to some other $U^{\prime}$ for which Eq.~\eqref{eq:Graalagain} holds.
\end{enumerate}
\end{Theorem}
Note that Eq.~\eqref{eq:Graalagain} is equivalent to
\begin{equation}\label{eq:Graaltris}
\Acal\cup\Bcal=\EcalKDCpu=\EcalKDCext,
\end{equation} 
where  $\EcalKDCpu$ denotes the set of pure KD-positive states.  
For general $\Acal$ and $\Bcal$, one has
\begin{equation}\label{eq:pureinext}
\Acal\cup\Bcal\subseteq \EcalKDCpu\subseteq \EcalKDCext.
\end{equation}
In other words, Eq.~\eqref{eq:Graalagain} corresponds to the simplest situation, where the set of extreme states is minimal.  In that case, we have  {an} explicit description of the set of all KD-positive states since the convex set $\convAB $ forms a  polytope with a simple geometric structure, detailed in Appendix A.  Note that part~(iv) of the theorem guarantees that the property Eq.~\eqref{eq:Graalagain} is stable in the sense that it is verified in an open set of unitary matrices. We conjecture that part (ii) of the theorem in fact holds in all dimensions. In other words, we think that the simple structure obtained in Eq.~\eqref{eq:Graalagain} is ``typically'' realized, meaning that it holds in an open dense set of full measure. We have numerically checked this  conjecture by randomly choosing unitary matrices $U$ for dimensions $d$ up to $10$ (See Section~\ref{s:graalproof} for details). The following proposition, proven in Section~\ref{s:graalproof}, shows a partial result in this direction:
\begin{Prop}\label{prop:pure_char} Let $d\geq2$.
There exists an open dense set of unitaries of probability $1$ for which $\EcalKDCpu=\Acal\cup\Bcal$.
\end{Prop} 
 We stress that we do nevertheless not know if, for the unitaries referred to in the proposition, the stronger property $\EcalKDC=\convAB$ holds.

In general, it is a formidable task to identify  $\EcalKDCpu$ and $\EcalKDCext$, given two specific bases $\Acal$ and $\Bcal$. Part~(iii) of the theorem shows that, when $U$ is the DFT matrix, and the dimension $d$ is a prime number, one again {satisfies} Eq.~\eqref{eq:Graaltris}. 
When $d$ is prime and { the columns of $U$ form two}  {mutually unbiased} (MUB) {bases with the canonical basis},  it is still true that $\EcalKDCpu=\Acal\cup\Bcal$ (See~\cite{Xu22} and Appendix~\ref{s:MUBpure}). But in that case, we have no information about the possible existence of mixed extreme states. 
When the dimension is not prime, and $U$ the DFT matrix, one can identify all pure KD-positive states~\cite{debievre2021,Xu22, debievre2023a} and one observes that there exist pure KD-positive states that are not basis states, \emph{i.e.} $\Acal\cup\Bcal\subsetneq \EcalKDCpu$. It is again, to the best of our knowledge, not known in that case if there also exist extreme KD-positive states that are mixed, meaning if $\EcalKDCpu\subsetneq \EcalKDCext$.

By analyzing in detail the situation where the transition matrix is real-valued, we  provide below (Section~\ref{s:eurekanew}) examples for which $
\Acal\cup\Bcal \subsetneq \EcalKDCpu\subsetneq \EcalKDCext$ or $\Acal\cup\Bcal = \EcalKDCpu\subsetneq \EcalKDCext. $
In these cases, there therefore exist mixed extreme states, some of which we explicitly identify. {We will highlight such situations with examples  where the bases $\Acal$ and $\Bcal$ are the eigenbases of two spin-$1$ components in some particular directions. }  While this situation is in a sense exceptional, {it has a precise analogue in the analysis of Wigner function.} Indeed, the pure Wigner positive  states are known to be the Gaussian states~\cite{hudson1974}. But it is also well known that the convex hull of the pure Gaussian states (which contains all Gaussian states) does not exhaust all Wigner positive states~\cite{genonietal2013}. As it turns out, even though examples of Wigner positive states not in this convex hull have been constructed~\cite{genonietal2013, vanherstraetencerf2021, hertzdebievre2023}, a complete description of the extreme states of the set of all Wigner positive states is, to the best of our knowledge, not available. In fact, no mixed extreme states have been explicitly identified for the Wigner function \cite{gross2006}.

 The {remainder of this} paper is structured as follows. In Section~\ref{s:setting}, we describe the general framework of our investigation, recall some definitions and necessary background information, and {introduce our}  notation. In Section~\ref{s:gen_res}, we  
prove several results on the {general structure of the} geometry of the set of KD-positive states. {These results} are essential ingredients for the proofs of our main results. 
In Section~\ref{s:graalproof}, we prove Theorem \ref{thm:Graal} and Proposition~\ref{prop:pure_char}. 
 In Section~\ref{s:eurekanew}, we focus on real unitary matrices to construct examples of mixed states that are KD positive but cannot be written as convex combinations of pure KD-positve states. Section~\ref{s:conc} contains our conclusions and outlook. 

\section{The setting and background}\label{s:setting}

In this section, we introduce some notation, define KD distributions and recall some of their  properties.

Throughout this manuscript, we consider a complex Hilbert space $\H$ of dimension $d$. We consider also two orthonormal bases $ \left(\cket{a_i}\right)_{i\in\IntEnt{1}{d}}$ and $\left(\cket{b_j}\right)_{j\in\IntEnt{1}{d}}$ in $\H$. We denote by $U=\left(\bracket{a_i}{b_j}\right)_{(i,j)\in\IntEnt{1}{d}^2}$ the transition matrix between these two bases. If $\rho$ is a density matrix, we define the KD distribution $Q(\rho)$  to be the $d \times d$ matrix~\cite{kirkwood1933, dirac1945}
\begin{equation}\label{eq:KDGV2}
\forall (i,j)\in \IntEnt{1}{d}^2, Q_{ij}(\rho) = \bracket{b_j}{a_i}\bra{a_i}\rho\cket{b_j}.
\end{equation}
Note that, for a given $\rho$, the matrix $Q(\rho)$ depends on the two bases. Although this will be crucial for our developments below, we do not indicate this dependence, not to burden the notation. 
The KD distribution thus satisfies the following   of Kolmogorov's axioms for joint probability distributions:
\begin{equation}\label{eq:Qbasics}
\sum_{i,j} Q_{ij}(\rho)=\Tr \rho=1, \quad \sum_j Q_{ij}(\rho)=\langle a_i|\rho|a_i\rangle,\quad \sum_i Q_{ij}(\rho)=\langle b_j|\rho|b_j\rangle.
\end{equation}

However, unlike joint probabilities, $Q(\rho)$ is in general a complex-valued matrix. {We call} a state KD positive whenever  $Q_{ij}(\rho)\geq 0$ for all $(i,j)\in\IntEnt{1}{d}^2$. The transition matrix $U$ is determined by the choice of bases and  determines whether $\rho$ is KD positive. For example, if $U = \bbone_{d}$, the bases are identical. Then, clearly all states are KD positive. 

We will say that two bases, $(\cket{a_i})_i$ and $(\cket{a_i'})_i$ or $(\cket{b_j})_j$ and $(\cket{b_j'})_j$, are equivalent if they can be obtained from each other by permutations of the basis vectors and/or phase rotations. In that case, the matrices $U$ and $U'$ are obtained from one another by permutations of their columns and rows, and global  phase rotations of the rows and columns. We shall say such matrices are equivalent. The point of these definitions is that, when the bases (and hence the transition matrices) are equivalent, then the corresponding sets of KD-positive states are identical. 
In particular, we note for later use that, if $\phi=(\phi_1,\dots, \phi_d)\in [0,2\pi)^d$, $\psi=(\psi_1,\dots,\psi_d)\in[0,2\pi)^d$, and if 
$$
|a_j'\rangle=\exp(-i\phi_j)|a_j\rangle,\quad |b_j'\rangle=\exp(-i\psi_j)|b_j\rangle,
$$
then the transition matrix $U'_{ij}=\langle a_i'|b_j'\rangle$ is given by
\begin{equation}\label{eq:phases}
U'=D(-\phi)UD(\psi),
\end{equation}
where, for any $\phi=(\phi_1,\dots, \phi_d)\in [0,2\pi)^d$,
$$
D(\phi)_{jk}=\exp(-i\phi_j)\delta_{jk}.
$$

Let us point out that we will often identify a unit vector $|\psi\rangle\in\Hcal$ with its projector $|\psi\rangle\langle\psi|$.

The questions we address in this work are of interest only if the two bases are in a suitable sense incompatible. 
For most of this work we will therefore assume that the unitary matrix $U$ has no zeros:
\begin{equation}\label{eq:mab}
\mab=\min_{i,j}|\langle a_i|b_j\rangle|>0.
\end{equation}
This guarantees that $Q(\rho)$ determines a unique $\rho$ (see Eq.~\eqref{eq:infcomplete}). In addition, it implies that none of the $|a_i\rangle\langle a_i|$ commutes with any of the $|b_j\rangle\langle b_j|$.  This is a weak  form of incompatibility between the two bases~\cite{debievre2023a}. Indeed, {$\mab > 0$ means that} if a measurement in the $\Acal$ basis yields an outcome $i$, then a subsequent measurement in the $\Bcal$ basis may yield any outcome $j$ with a nonvanishing probability.  
We recall that a special role is played by mutually unbiased (MUB) bases, for which $\mab$ takes the maximum possible value $\mab=\frac{1}{\sqrt{d}}$. All outcomes $j$ for a $\Bcal$-measurement after an initial measurement in the $\Acal$-basis are then equally probable, and vice versa. 

\section{General structural results}\label{s:gen_res}
In this section, we prove general results regarding the geometry of KD-positive states. We work under the  assumption that $\mab>0$.

 \subsection{The KD symbol of an observable}
 It is well known that the Wigner function can be defined not only for states $\rho$, but also for arbitrary observables $F$, in which case it is referred to as the Weyl symbol of $F$. One can proceed similarly {with} the KD distribution.  
 Denoting by $\SAO$ the set of self-adjoint operators,  we define, 
 \begin{equation} \label{eq:Qdef}
Q : \left\{\begin{array}{rcl}
\SAO  & \rightarrow &\MatC{d} \\
 F & \mapsto & (Q_{ij}( F))_{(i,j)\in\IntEnt{1}{d}^2}
\end{array}\right.,
\end{equation}
where
\[
\forall (i,j)\in\IntEnt{1}{d}^2, \  Q_{ij}(F)=\langle a_i|F|b_j\rangle \langle b_j|a_i\rangle ,
\]
and where $\MatC{d}$ is the the space of {complex} $d$ by $d$  matrices.
We shall refer to $Q(F)$ as the KD symbol of $F$. 
We note that
\begin{equation}\label{eq:Qmarginals}
\sum_{j=1}^{d} Q_{ij}(F)=\langle a_i|F|a_i\rangle\in\R, \quad \sum_{i=1}^{d} Q_{ij}(F)=\langle b_j|F|b_j\rangle\in\R, \quad \sum_{i,j}Q_{ij}(F)=\Tr (F).
\end{equation}
Also, for $F,G\in\SAO$, we have 
\[
\begin{array}{rcl}
\mathrm{Tr}(FG) &=&  \sum_{(i,j)\in\IntEnt{1}{d}^2} \bra{a_i}F\cket{b_j}\bra{b_j}G\cket{a_i} \\
&=& \displaystyle\sum_{(i,j)\in\IntEnt{1}{d}^2} \frac{1}{\left|\bracket{a_i}{b_j}\right|^2}\bracket{b_j}{a_i}\bra{a_i}F\cket{b_j}\bracket{a_i}{b_j}\bra{b_j}G\cket{a_i} \\
&=&\displaystyle \sum_{(i,j)\in\IntEnt{1}{d}^2}\frac{1}{\left|\bracket{a_i}{b_j}\right|^2}Q_{ij}(F)\overline{Q_{ij}(G)}.
\end{array}\]
If $\Acal$ and $\Bcal$ are MUB bases, then
\[
\mathrm{Tr}(FG) = d\sum_{(i,j)\in\IntEnt{1}{d}^2}Q_{ij}(F)\overline{Q_{ij}(G)}=d \mathrm{Tr}(Q(F)Q^{\dagger}(G)).
\]
One may note the analogy between these two identities and the well known ``overlap identity'' for the Wigner function/Weyl symbol which expresses $\Tr(FG)$ as a phase space integral of the product of the Wigner function/Weyl symbol of $F$ and $G$~\cite{leonhardt2010}.

We point out that, when $\mab>0$, the KD symbol $Q(F)$ determines the observable $F$ uniquely. The reconstruction formula is~\cite{johansen2007}
\begin{equation}\label{eq:infcomplete}
\forall (i,k)\in\IntEnt{1}{d}^2, \ \langle a_i|F|a_k\rangle = \sum_{j=1}^{d} Q_{ij}(F)\frac{\langle b_j|a_k\rangle}{\langle b_j|a_i\rangle}.
\end{equation}
This property is sometimes referred to as informational completeness. 
In other words, the map $Q$ is injective: $\Ker Q=\{0\}$ if $\mab>0${, where $\Ker Q$ denotes the kernel of $Q$.}

Since  the dimension of the real vector space $\SAO$ is $d^2$, it follows that
$\dim \Ran Q= d^2,$ where $\Ran{Q}$ denotes the image of {$Q$}. Hence, $\Ran Q$ is a $d^2$-dimensional real vector subspace of the $2d^2$-dimensional real vector space $\MatC{d}$. Note that a matrix $M\in\mathcal \blue{\MatC{d}}$  belongs to $\Ran Q$ if and only if it satisfies the $d^2$ real linear constraints 
\[
\sum_j M_{kj}\frac{\langle b_j|a_i\rangle}{\langle b_j|a_k\rangle}=\sum_j\overline{M}_{ij}\frac{\langle a_k|b_j\rangle}{\langle a_i|b_j\rangle}.
\]
We will further find it useful to consider the imaginary part of $Q$:
\[
\Im{Q}: \left\{\begin{array}{rcl}
\SAO  & \rightarrow &\MatR{d} \\
 F & \mapsto & (\Im{Q_{ij}(F)})_{(i,j)\in\IntEnt{1}{d}^2}
\end{array}\right.,
 \]
which is a real-linear map into the space of real matrices $\MatR{d}$. 

To streamline the discussion, we introduce the following terminology. We will say that a self-adjoint operator $F$ is a \emph{KD-real operator}
whenever its KD distribution is real-valued. In other words, $F$ is KD real if and only if 
\[
F\in\VR:=\Ker\,{\mathrm{Im}}Q=Q^{-1}(\mcl{M}_{d}(\R)).
\]
 We will say it is \emph{KD positive}  
 if its KD distribution takes on real nonnegative values only. In other words, iff 
 \[
 F\in \VRp:=Q^{-1}(\mcl{M}_{d}(\R^{+})) \subseteq \VR.
 \]
 We point out for later use that, since $\VR= \Ker(\mathrm{Im} Q)\subset \SAO$ and $\dim\SAO=d^2$, 
 \begin{equation}
 \dim \VR \leq d^2.
 \end{equation}
Clearly,  if $F_1, F_2\in\VRp$, then $\lambda_1 F_1+\lambda_2F_2\in\VRp$ for all $\lambda_1, \lambda_2\geq 0$. In particular,
$\VRp$ is a closed convex cone. Note that it has no extreme points, except for the origin.

\subsection{The case $\EcalKDC=\convAB$: a geometric condition}
Recall that a density matrix representing a quantum state  is a nonnegative operator $\rho$ satisfying $\mathrm{Tr}\rho=1$. We will write $\SAOpone$ for the set of density matrices and $\SAOp$ for the set of positive operators.  Hence 
\begin{equation}\label{eq:EcalKDCchar}
\EcalKDC=\VRp\cap \SAOpone.
\end{equation}
Note that $\EcalKDC$ is compact
so that, by the Krein-Milman theorem,
$
\EcalKDC=\conv{\EcalKDCext}.
$

 The question we are addressing in this section 
 is under which conditions on $U$ it is true that
\begin{equation}\label{eq:Graal}
\EcalKDC=\convAB,
\end{equation}
where 
\begin{eqnarray}
\convAB&=&\conv{\{\cket{a_i}\bra{a_i}, \cket{b_j}\bra{b_j} \mid 1\leq i,j\leq d\}}.
\end{eqnarray}
 In other words, the question is:  Is it  true or false that all KD-positive states are convex mixtures of the basis states? 
 This is equivalent to {checking if} the inclusions in Eq.~\eqref{eq:pureinext} are equalities, { \textit{i.e.}, if}
\begin{equation}\label{eq:Graalequiv}
\Acal\cup\Bcal=\EcalKDCpu=\EcalKDCext.
\end{equation}
One can think of Eq. \eqref{eq:Graalequiv} as the situation where the set $\EcalKDC$ of KD-positive states is the smallest possible.   In a sense then, this corresponds to the choice of two bases $\Acal$ and $\Bcal$ that are ``most strongly quantum.'' Note that, when Eq. \eqref{eq:Graalequiv} holds true,
$\EcalKDC$ is a convex polytope with 2d known summits $\{\cket{a_i}\bra{a_i}, \cket{b_i}\bra{b_i}\}_{i\in\IntEnt{1}{d}}$. Its geometry is described in  Appendix~\ref{s:convABstructure}.  

In Proposition~\ref{prop:TSKD2bis} we will prove conditions of a geometric nature on the set of KD-real operators that are equivalent to Eq.~\eqref{eq:Graal}. 

 We introduce the vector space
\begin{eqnarray}
\spanRAB&=&\mathrm{span}_{\R} \{\cket{a_i}\bra{a_i}, \cket{b_j}\bra{b_j} \mid 1\leq i,j\leq d\}.
\end{eqnarray}
and show the following result:
\begin{Lemma}\label{lem:G1}
If $\mab> 0$, then
\[
\mathrm{dim}\left(\spanRAB\right) =2d- 1\quad\mathrm{and}\quad \EcalKDC\cap \spanRAB =\convAB.
\]
\end{Lemma}

\begin{proof} To prove the first statement, we consider the linear map 
\[
\Gamma : \left\{\begin{array}{rcl}
 \R^{2d}  & \rightarrow & \spanRAB \\
 ((\lambda_{i})_{i\in\IntEnt{1}{d}},(\mu_{j})_{j\in\IntEnt{1}{d}})& \mapsto & \sum_{i=1}^{d} \lambda_{i}\cket{a_i}\bra{a_i} + \sum_{j=1}^{d} \mu_{j}\cket{b_j}\bra{b_j}
\end{array}\right.
\]
for which the rank theorem gives $\text{dim}\left(\spanRAB\right) = 2d- \text{dim}\left(\Ker(\Gamma)\right)$.
Now suppose that $F = \sum_{i=1}^{d} \lambda_{i}\cket{a_i}\bra{a_i} + \sum_{j=1}^{d} \mu_{j}\cket{b_j}\bra{b_j} =0$. Then
\[
\forall (i,j)\in\IntEnt{1}{d}^2, \bra{a_i}F\cket{b_j} = \bracket{a_i}{b_j}\left(\lambda_{i}+\mu_{j}\right)=0.
\]
Now, for any  $i\in\IntEnt{1}{d}$,  $\bracket{a_i}{b_1}\neq 0$ because $\mab > 0$, hence $\lambda_{i}+\mu_{1} = 0$ and finally $\lambda_{i} = -\mu_{1}$ for all $i\in\IntEnt{1}{d}$. 
Exchanging the roles $(\lambda_{i})_{i\in\IntEnt{1}{d}}$ and $(\mu_{j})_{j\in\IntEnt{1}{d}}$, we find that for all $j\in\IntEnt{1}{d}$,  $\mu_{j}=\mu_{1}$. So, the relation stands as
\[
\mu_{1}\left(-\sum_{i=1}^{d} \cket{a_i}\bra{a_i} + \sum_{j=1}^{d}\cket{b_j}\bra{b_j}\right) =0,
\]
which is true for all $\mu_{1}\in\R$. This means that $\text{dim}\left(\Ker(\Gamma)\right) =1$ and so 
\[
\text{dim}\left(\spanRAB\right) = 2d- 1.
\]

We now turn to the second statement. That $\convAB\subseteq\EcalKDC\cap \spanRAB$ is immediate. {Thus,} we only need to prove the other inclusion. 
Let therefore $\rho\in \EcalKDC\cap \spanRAB$. Hence, there exist $\lambda_i, \mu_j\in\R$ so that 
$$
\rho=\sum_{i=1}^{d} \lambda_i\cket{a_i}\bra{a_i} + \sum_{j=1}^{d} \mu_j\cket{b_j}\bra{b_j}.
$$ 
Consequently 
\[
\forall (i,j)\in\IntEnt{1}{d}^2, Q_{ij}(\rho) = \left|\bracket{a_i}{b_j}\right|^2 \left(\lambda_{i}+\mu_{j}\right).
\]
After a possible reordering of the basis, we can suppose that $\frac{Q_{11}(\rho)}{|\langle a_1|b_1\rangle|^2} = \min_{(i,j)\in\IntEnt{1}{d}^2}\frac{Q_{ij}(\rho)}{\left|\bracket{a_i}{b_j}\right|^2}$ so that 
\[
\forall j \in \IntEnt{2}{d}, \mu_j -\mu_1 =  \frac{Q_{1j}(\rho)}{\left|\bracket{a_1}{b_j}\right|^2} - \frac{Q_{11}(\rho)}{\left|\bracket{a_1}{b_1}\right|^2} \geqslant 0 \ \mathrm{and} \ \lambda_j -\lambda_1=\frac{Q_{j1}(\rho)}{\left|\bracket{a_j}{b_1}\right|^2} - \frac{Q_{11}(\rho)}{\left|\bracket{a_1}{b_1}\right|^2} \geqslant 0.
\]
Moreover, since $\rho$ is KD positive, $Q_{11}(\rho)=  \left|\bracket{a_1}{b_1}\right|^2 \left(\lambda_{1}+\mu_{1}\right) \geqslant 0$. So, either  $\mu_1$ or $\lambda_1$ must be nonnegative. Suppose $\lambda_1\geq 0$, then as $\cket{a_1}\bra{a_1} = \sum_{j=1}^{d}\cket{b_j}\bra{b_j} - \sum_{i=2}^{d}\cket{a_i}\bra{a_i}$ we can rewrite $\rho$ as 
\[
\rho = \sum_{i=2}^{d} \left(\lambda_i -\lambda_1\right)\cket{a_i}\bra{a_i} + \sum_{j=1}^{d} \left(\mu_j +\lambda_1\right)\cket{b_j}\bra{b_j} = \sum_{i=2}^{d} \left(\lambda_i -\lambda_1\right)\cket{a_i}\bra{a_i} + \sum_{j=1}^{d} \frac{Q_{1j}(\rho)}{\left|\bracket{a_1}{b_j}\right|^2}\cket{b_j}\bra{b_j}.
\]
Hence $\rho\in \mathrm{span}_{\R^+} (\Acal \cup \Bcal)$. Together with the fact that $\Tr\rho=1$, this shows that $\rho\in \convAB$ and {finalizes our}  proof.
\end{proof}
 
Recall that $\VR=\Ker(\Im{Q})$ and that 
$$
\spanRAB\subset \VR.
$$
So we conclude that 
\begin{equation}\label{eq:dimVRbounds}
2d-1\leq \dim \VR=\dim(\Ker(\Im{Q}))\leq d^2.
\end{equation}

The following proposition shows that the condition $\dim\VR=2d-1$ is equivalent to the requirement that the basis states are the only extreme KD-positive  states which is equivalent to Eq.~\eqref{eq:Graal}.\\

\begin{Prop}\label{prop:TSKD2bis}
Suppose $m_{A,B} > 0$. Consider the following statements:\\
(ia)  $\VR=\spanRAB $;\\
(ib) $\dim \VR=2d-1$;\\
(iia) $\EcalKDCext= \Acal \cup \Bcal$;\\
(iib) $\EcalKDC=\convAB$.\\
Then (ia) $\Leftrightarrow$ (ib) $\Leftrightarrow$ (iia) $\Leftrightarrow$ (iib). 
\end{Prop}
\begin{proof} That (ia) $\Leftrightarrow$ (ib) is immediate and so is the equivalence between (iia) and (iib).

We first show that (ia) implies (iib). Let $\rho\in\EcalKDC$. Then it belongs to $\VR$ and hence, by (ia), $\rho\in\spanRAB$. Hence, by the second statement of Lemma~\ref{lem:G1}, as $\rho\in \EcalKDC\cap \spanRAB$, it follows that $\rho\in\convAB$. Thus, $\EcalKDC=\convAB$.

It remains to show that (iib) implies (ia).
We proceed by contraposition. Suppose that (ia) does not hold so that $\dim \VR>2d-1$. Lemma~\ref{lem:G1} then implies that 
$\spanRAB$ is a proper subspace of $\VR$. So
\[
\VR=\spanRAB \oplus W,
\]
with $W$ equal to the orthogonal complement of $\spanRAB$ in $\VR$, which is nontrivial by assumption. Note that $F\in W$ implies that $\langle a_i|F|a_i\rangle=0=\langle b_j|F|b_j\rangle$ for all $(i,j)\in  \IntEnt{1}{d}^2$. This implies $\Tr F=0$.
In addition, $Q(F)$ has only real entries by the definition of $\VR$.  Choose $F\in W\backslash \{0\}$, 
and consider, for all $x\in\R$
\[
\rho(x)=\rho_*+xF \ \mathrm{where} \ 
\rho_{*}=\frac1{d} \bbone_d\in \convAB.
\]
Note that $\Tr \rho(x)=1$ for all $x\in\R$ and that, for all $x\in\R$, one has
$$
\langle \psi|\rho(x)|\psi\rangle \geq \frac1{d}-f_{\mathrm{max}}|x| ,
$$
where $\cket{\psi}$ is any norm-1 vector in $\Hcal$.
Here, $f_{\mathrm{max}}=\max\{|f_i||i\in \IntEnt{1}{d}\}>0$, where the $f_i$ are the eigenvalues of $F$. In particular,
if $|x|\leq\frac1{df_{\mathrm{max}}}$, then $\rho(x)$ is a positive operator of trace $1$. 
We now show that there exist $0< x_+\leq\frac1{d f_{\mathrm{max}}}<+\infty$  so that
\begin{equation}\label{eq:xsmall}
\forall x\in [-x_+,x_+],\quad \rho(x)\in \EcalKDC.
\end{equation}
Since $F\in \VR$, we know $\rho(x)\in\VR$. 
One has, for all $x\in\R$,
$$
Q_{i j}(\rho(x))=\left|\bracket{a_i}{b_j}\right|^2+x Q_{ij}(F)\geq \mab^2-|x|\max_{i,j}|Q_{i j}(F)|.
$$
Taking $x_+=\min\{\frac1{d f_{\mathrm{max}}}, \frac{\mab^2}{\max_{i,j}|Q_{i j}(F)|}\} > 0$, we have Eq.~\eqref{eq:xsmall}.
This implies that (iib) does not hold since for all $x\not=0$, $\rho(x)\not\in\spanRAB$. 
\end{proof}

\subsection{Characterizing $\rho \in \convAB$ }

The following proposition is essential to the proof of Theorem~\ref{thm:Graal}~(iii). 

\begin{Prop}\label{prop:CSC2}
Suppose $\mab > 0$. Then,
$$
\rho\in \convAB
$$
 if and only if $\rho\in\EcalKDC$ and
\begin{equation}\label{eq:supcond2}
\quad \forall (i,j,k,l)\in\IntEnt{1}{d}^4,  \frac{Q_{ij}(\rho)}{\left|\bracket{a_i}{b_j}\right|^2} + \frac{Q_{kl}(\rho)}{\left|\bracket{a_k}{b_l}\right|^2} =  \frac{Q_{il}(\rho)}{\left|\bracket{a_i}{b_l}\right|^2} +\frac{Q_{kj}(\rho)}{\left|\bracket{a_k}{b_j}\right|^2}.
\end{equation}
\end{Prop}

\begin{proof}
We first show the reverse implication. 
Let $\rho \in \EcalKDC$ and satisfying~Eq.~\eqref{eq:supcond2}. We construct a state $\rho_2 \in \convAB$ such that $Q(\rho_2) = Q(\rho)$. Since $\mab>0$, we know from~Eq.~\eqref{eq:infcomplete} that the KD distribution determines the state, such that $\rho_2 = \rho$.

Note that the basis states have the following KD distribution: 
\[
\forall (i,j,k)\in\IntEnt{1}{d}^3, Q_{ij}(\cket{a_k}\bra{a_k}) = \left|\bracket{a_k}{b_j}\right|^2\delta_{i,k} \text{ and } Q_{ij}(\cket{b_k}\bra{b_k}) = \left|\bracket{a_i}{b_k}\right|^2 \delta_{j,k}.
\]
By permuting the order of the vectors in $\mcl{B}$, we can suppose that $\frac{Q_{11}(\rho)}{\left|\bracket{a_1}{b_1}\right|^2} =  \min_{j\in\IntEnt{1}{d}} \frac{Q_{1j}(\rho)}{\left|\bracket{a_1}{b_j}\right|^2}$. We define 
\[
\rho_{2} = \sum_{i=1}^{d}\lambda_{i}\cket{a_i}\bra{a_i} + \sum_{j=1}^{d}\mu_{j}\cket{b_j}\bra{b_j}
\]
where $\lambda_{i} = \frac{Q_{i1}(\rho)}{\left|\bracket{a_i}{b_1}\right|^2}$ for all $i\in\IntEnt{1}{d}$ and $\mu_{j} = \frac{Q_{1j}(\rho)}{\left|\bracket{a_1}{b_j}\right|^2}-\frac{Q_{11}(\rho)}{\left|\bracket{a_1}{b_1}\right|^2}$ for $j\in\IntEnt{1}{d}$ so that $\mu_{1}=0$.
Since $\rho$ is KD positive, $\lambda_{i}\geqslant 0$ and $\mu_{i}\geqslant 0$ for all $i\in\IntEnt{1}{d}$. Moreover, using Eq.~\eqref{eq:supcond2}, one has that
\[
\begin{array}{rcl}
\displaystyle \Tr \rho_2=\sum_{j=1}^{d} \mu_{j} +\sum_{i=1}^{d} \lambda_{i} &=&  \sum_{j=1}^{d} \sum_{i=1}^{d} \left|\bracket{a_i}{b_j}\right|^2\mu_{j} + \sum_{i=1}^{d} \sum_{j=1}^{d} \left|\bracket{a_i}{b_j}\right|^2 \lambda_{i}  \\
&=& \displaystyle \sum_{j=1}^{d} \sum_{i=1}^{d} \left|\bracket{a_i}{b_j}\right|^2\left(\frac{Q_{i1}(\rho)}{\left|\bracket{a_i}{b_1}\right|^2} + \frac{Q_{1j}(\rho)}{\left|\bracket{a_1}{b_j}\right|^2}-\frac{Q_{11}(\rho)}{\left|\bracket{a_1}{b_1}\right|^2} \right) \\
&=& \displaystyle \sum_{j=1}^{d} \sum_{i=1}^{d} Q_{ij}(\rho) = 1,
\end{array} 
\]
so that $\rho_2\in\convAB$.  
Using Eq.~\eqref{eq:supcond2} again, we find $\forall (i,j)\in\IntEnt{1}{d}^2,$
\[
 \begin{array}{rcl}
 Q_{ij}(\rho_{2}) &=& \left|\bracket{a_i}{b_j}\right|^2\left(\lambda_{i}+\mu_{j}\right) \\
 & = &  \left|\bracket{a_i}{b_j}\right|^2\left( \frac{Q_{i1}(\rho)}{\left|\bracket{a_i}{b_1}\right|^2}+ \frac{Q_{1j}(\rho)}{\left|\bracket{a_1}{b_j}\right|^2}-\frac{Q_{11}(\rho)}{\left|\bracket{a_1}{b_1}\right|^2}\right) \\
 &=& Q_{ij}(\rho).
 \end{array}
\]
This shows that $\rho_2=\rho$ so that $\rho\in\convAB$. 

For the proof of the direct implication, we note that if $\rho\in\convAB$, then $\rho=\sum_{i=1}^{d}\lambda_{i}\cket{a_i}\bra{a_i} + \sum_{j=1}^{d}\mu_{j}\cket{b_j}\bra{b_j}$ with $\sum_{i=1}^{d} \lambda_{i}+\mu_{i} = 1$, $\lambda_{i}\geqslant 0$ and $\mu_{i}\geqslant 0$ for all $i\in\IntEnt{1}{d}$. The KD distribution of $\rho$ is given by
\[
\forall (i,j)\in\IntEnt{1}{d}^2, Q_{ij}(\rho) = \left|\bracket{a_i}{b_j}\right|^2\left(\lambda_{i}+\mu_{j}\right).
\]
Hence, $\rho$ is KD-positive and  for all $ (i,j)\in\IntEnt{1}{d}^2$, 
\[
\frac{Q_{11}(\rho)}{\left|\bracket{a_1}{b_1}\right|^2} + \frac{Q_{ij}(\rho)}{\left|\bracket{a_i}{b_j}\right|^2} = \left(\lambda_{1}+\mu_{1}\right) +\left(\lambda_{i}+\mu_{j}\right) = \lambda_{1}+\mu_{1} + \lambda_{i}+\mu_{j} ,
\] 
\[
\frac{Q_{1j}(\rho)}{\left|\bracket{a_1}{b_j}\right|^2} + \frac{Q_{i1}(\rho)}{\left|\bracket{a_i}{b_1}\right|^2} = \left(\lambda_{1}+\mu_{j}\right) +\left(\lambda_{i}+\mu_{1}\right) = \lambda_{1}+\mu_{1} + \lambda_{i}+\mu_{j} .
\]
This implies Eq.~\eqref{eq:supcond2} with $k=1=l$. For the general case,
we write
\[
\frac{Q_{ij}(\rho)}{\left|\bracket{a_i}{b_j}\right|^2} + \frac{Q_{kl}(\rho)}{\left|\bracket{a_k}{b_l}\right|^2} =  \frac{Q_{i1}(\rho)}{\left|\bracket{a_i}{b_1}\right|^2} +\frac{Q_{1j}(\rho)}{\left|\bracket{a_1}{b_j}\right|^2} -  \frac{Q_{11}(\rho)}{\left|\bracket{a_1}{b_1}\right|^2} +\frac{Q_{k1}(\rho)}{\left|\bracket{a_k}{b_1}\right|^2} +\frac{Q_{1l}(\rho)}{\left|\bracket{a_1}{b_l}\right|^2} -  \frac{Q_{11}(\rho)}{\left|\bracket{a_1}{b_1}\right|^2}
\]
and
\[
\frac{Q_{il}(\rho)}{\left|\bracket{a_i}{b_l}\right|^2} + \frac{Q_{kj}(\rho)}{\left|\bracket{a_k}{b_j}\right|^2} =  \frac{Q_{i1}(\rho)}{\left|\bracket{a_i}{b_1}\right|^2} +\frac{Q_{1l}(\rho)}{\left|\bracket{a_1}{b_l}\right|^2} -  \frac{Q_{11}(\rho)}{\left|\bracket{a_1}{b_1}\right|^2} +\frac{Q_{k1}(\rho)}{\left|\bracket{a_k}{b_1}\right|^2} +\frac{Q_{1j}(\rho)}{\left|\bracket{a_1}{b_j}\right|^2} -  \frac{Q_{11}(\rho)}{\left|\bracket{a_1}{b_1}\right|^2}.
\]
The right hand sides of these two equations are identical up to a reorganization of the terms, so
\[
\frac{Q_{ij}(\rho)}{\left|\bracket{a_i}{b_j}\right|^2} + \frac{Q_{kl}(\rho)}{\left|\bracket{a_k}{b_l}\right|^2} = \frac{Q_{il}(\rho)}{\left|\bracket{a_i}{b_l}\right|^2} + \frac{Q_{kj}(\rho)}{\left|\bracket{a_k}{b_j}\right|^2}.
\]
This ends {our} proof.
\end{proof}

The relations~\eqref{eq:supcond2} are simpler for MUB bases and are given in the following corollary. 
\begin{Cor}\label{cor:CSC1}
Let $\mcl{A}$ and $\mcl{B}$ be MUB bases.
Then
$$
\rho\in \convAB
$$ 
if and only if $\rho\in\EcalKDC$ and 
\begin{equation}\label{eq:supcond4}
\quad \forall (i,j,k,l)\in\IntEnt{1}{d}^4, Q_{ij}(\rho) +Q_{kl}(\rho) = Q_{il}(\rho) +Q_{kj}(\rho).
\end{equation}
\end{Cor}
\begin{proof}
This is a direct consequence of Proposition~\ref{prop:CSC2}. 
\end{proof}

\section{Proofs of Theorem~\ref{thm:Graal} and of Proposition~\ref{prop:pure_char}.}\label{s:graalproof}
For convenience, we restate {our} theorem:
\begin{thmIntro}
The equality
\begin{equation}\tag{\ref{eq:Graalagain}}
\EcalKDC=\convAB,
\end{equation}
holds under any single one of the following hypotheses:
\begin{enumerate}[label=(\roman*),wide, labelindent=0pt]
    \item If $d=2$ (for qubits) and $\mab>0$;
    \item If $d=3$, for all $U$ in a open dense set of probability $1$;
    \item If $d$ is prime and $U$ is the discrete Fourier transform (DFT) matrix;
    \item If $U$ is sufficiently close to some other $U^{\prime}$ for which Eq.~\eqref{eq:Graalagain} holds.
\end{enumerate}
\end{thmIntro}
\noindent\emph{Proof of Theorem~\ref{thm:Graal}~(i).}
We can, without loss of generality, suppose that the transition matrix $U$ is a real matrix by executing appropriate phase changes on the basis vectors. 
If  $U$ has no zeros ($\mab>0$), we can therefore write $U=\begin{pmatrix} \cos(\theta) & \sin(\theta) \\ -\sin(\theta) & \cos(\theta) \end{pmatrix}$ for $\theta\in\R\backslash\frac{\pi}{2}\Z$. To find the dimension of the space of KD-real operators, we consider $F\in\VR$ and write
\begin{eqnarray*}
Q_{11}(F) &=& \bracket{b_1}{a_1}\bra{a_1}F\cket{b_1} =  \bracket{b_1}{a_1}\bracket{a_1}{b_1}F_{11} +  \bracket{b_1}{a_1}\bracket{a_2}{b_1}F_{12}\\
&=&   \cos(\theta)^2F_{11} +  \cos(\theta)\sin(\theta)F_{12}\in\R,
\end{eqnarray*}
with $F_{ij}=\langle a_i|F|a_j\rangle$.
Since, by hypothesis, $F$ is self-adjoint and $\Im{Q_{1 1}(F)}=0$, one finds  $\Im{F_{12}}=0$ so $F_{12} = F_{21}$.
Hence $F\in \VR$ implies that $F$ is real symmetric. Conversely, one {can check} that for any real symmetric $F$, $Q(F)$ is a real matrix. Consequently, $\dim(\VR)=3$. The result then follows from Proposition~\ref{prop:TSKD2bis}.

\noindent \emph{Proof of Theorem~\ref{thm:Graal}~(ii).}
This result is restated more explicitly  in the following proposition. 

\begin{Prop}
In dimension $d=3$, there exists a set $\mathcal W$ of unitary matrices such that:
\begin{itemize}
    \item $\forall U \in \mathcal W$, $\EcalKDC = \convAB$; 
    \item $\mathcal{W}$ is an open and dense subset of the set of unitary matrices;
    \item $\mathcal W$ is a set of probability one for the Haar measure on the unitary group.
\end{itemize}
\end{Prop}
\begin{proof}
For any unitary matrix $U$ with $\mab > 0$, we write $U = (A_{kj}e^{i\phi_{kj}})_{(k,j)\in\IntEnt{1}{3}^2}$, with $A_{kj} > 0$. We define $\mathcal W$ to be the set of unitary matrices in dimension 3 for which $\mab > 0$ and the following conditions are fulfilled: 
\begin{equation}\label{cond:rk4}
\left\{\begin{array}{rclc}
\phi_{21}-\phi_{11}& \neq &0& [\frac{\pi}{2}] \\
\phi_{22}-\phi_{12} &\neq& 0& [\frac{\pi}{2}] \\
\phi_{31}-\phi_{11}& \neq& 0& [\frac{\pi}{2}] \\
\phi_{32}-\phi_{12} &\neq &0& [\frac{\pi}{2}] \\
\phi_{21}-\phi_{11} &\neq &\phi_{22}-\phi_{12}& [\pi] \\
\phi_{31}-\phi_{11} &\neq &\phi_{32}-\phi_{12} & \ [\pi].
\end{array}\right. 
\end{equation}
 Let $U\in\mathcal W$. We want to show that $\EcalKDC = \convAB$. According to Proposition~\ref{prop:TSKD2bis}, it is sufficient to show that $\dim(\Ran{\Im{Q}})=4$. (Here and below, $\Ran{T}$ stands for the range of  a linear map $T$.) For that purpose, we shall consider the $9\times 9$ matrix $T$ of the linear map $\mathrm{Im}Q: \SAO   \rightarrow \MatR{d} $ with respect to the basis 
$$
\left\{\cket{a_k}\bra{a_k}, (\cket{a_k}\bra{a_j}+\cket{a_j}\bra{a_k}), i(\cket{a_k}\bra{a_j}-\cket{a_j}\bra{a_k})\right\}_{k\in\IntEnt{1}{3},k< j}
$$ 
of $\SAO$ and the canonical basis of $\MatR{d}$. The matrix $T$ can be readily computed but we do not display it here.  Note that, by Eq.~\eqref{eq:dimVRbounds}, $\dim(\Ker{(\Im{Q}})) \geqslant 5$, so that $\dim{(\Ran{\Im{Q}})} = 9 - \dim(\Ker{(\Im{Q}})) \leqslant 4$. Equality is obtained, \emph{i.e.} $\dim{(\Ran{\Im{Q}})}= 4$, if and only if  there exists a $4$ by $4$ submatrix $\Sigma$ of $T$ that has rank $4$. We will show that the submatrix $\Sigma$, given by 
\[
\Sigma = \begin{psmallmatrix}
A_{11}A_{21}\sin{(\phi_{21}-\phi_{11})}&A_{11}A_{21}\cos{(\phi_{21}-\phi_{11})}&A_{11}A_{31}\sin{(\phi_{31}-\phi_{11})}&A_{11}A_{31}\cos{(\phi_{31}-\phi_{11})}\\
A_{12}A_{22}\sin{(\phi_{22}-\phi_{12})}&A_{12}A_{22}\cos{(\phi_{22}-\phi_{12})}&A_{12}A_{32}\sin{(\phi_{32}-\phi_{12})}&A_{12}A_{32}\cos{(\phi_{32}-\phi_{12})}\\
-A_{11}A_{21}\sin{(\phi_{21}-\phi_{11})}&-A_{11}A_{21}\cos{(\phi_{21}-\phi_{11})}&0&0\\
-A_{12}A_{22}\sin{(\phi_{22}-\phi_{12})}&-A_{12}A_{22}\cos{(\phi_{22}-\phi_{12})}&0&0
\end{psmallmatrix},
\]
is indeed of rank $4$. To prove this, suppose there exists $(a_1,a_2,a_3,a_4)\in\R^{4}$ such that $(a_1,a_2,a_3,a_4)\in\Ker{\Sigma}$. Then, 
\[
\left\{\begin{array}{rcl}
A_{21}(a_1\sin{(\phi_{21}-\phi_{11})}+a_2\cos{(\phi_{21}-\phi_{11})})&=&-A_{31}(a_3\sin{(\phi_{31}-\phi_{11})}+a_4\cos{(\phi_{31}-\phi_{11})})\\
A_{22}(a_1\sin{(\phi_{22}-\phi_{12})}+a_2\cos{(\phi_{22}-\phi_{12})})&=& -A_{32}(a_3\sin{(\phi_{32}-\phi_{12})}+a_4\cos{(\phi_{32}-\phi_{12})})\\
a_1\sin{(\phi_{21}-\phi_{11})}+a_2\cos{(\phi_{21}-\phi_{11})} &=& 0\\
a_1\sin{(\phi_{22}-\phi_{12})}+a_2\cos{(\phi_{22}-\phi_{12})}&=&0.
\end{array}\right. 
\]
The last two rows simplify to
\[
\left\{\begin{array}{rcl}
a_1\tan{(\phi_{21}-\phi_{11})}&=& -a_2,\\
a_1\tan{(\phi_{22}-\phi_{12})}&=& -a_2.\\
\end{array}\right. 
\]
If $a_1\neq 0$, then $\tan{(\phi_{21}-\phi_{11})}=\tan{(\phi_{22}-\phi_{12})}$, which contradicts the condition $\phi_{21}-\phi_{11} \neq \phi_{22}-\phi_{12} [\pi]$. So $a_1=a_2=0$.
Consequently, the first two conditions reduce to
\[
\left\{\begin{array}{rcl}
a_3A_{11}A_{31}\sin{(\phi_{31}-\phi_{11})}+a_4A_{31}A_{11}\cos{(\phi_{31}-\phi_{11})}=0,\\
a_3A_{12}A_{32}\sin{(\phi_{32}-\phi_{12})}+a_4A_{32}A_{12}\cos{(\phi_{32}-\phi_{12})}=0.\\
\end{array}\right. 
\]
Following the same argument, we find that $a_3=a_4=0$.  
Consequently, the matrix $\Sigma$ has a vanishing kernel and is therefore of rank $4$. 
In conclusion, for any unitary matrix in $\mathcal W$ it is true that $\dim{(\Ran{\Im{Q}})}= 4$, and hence $\dim{\Ker{(\Im{Q}})}=5$. This concludes the proof of the first part of the Proposition.

The set $\mathcal W$ is clearly open.
We now show that it is dense {also}. For that purpose, consider an arbitrary unitary matrix $U$. Suppose it does not belong to $\mathcal W$ so that at least one  of the six conditions in~Eq.~\eqref{cond:rk4} is not satisfied for $U$. We write $C_1, C_2, C_3$ for the columns of $U$ and remark that $C_{3} = \epsilon C_{1}\wedge C_{2}$ with $\epsilon\in \{-1,1\}$; here $\wedge$ {denotes}  the vector product. We then construct, for $\theta\in\R$, the two columns
$$
C_{1}(\theta)=\begin{pmatrix} A_{11}e^{i(\phi_{11}+\theta)} \\ A_{21}e^{i\phi_{21}} \\ A_{31}e^{i\phi_{31}}\end{pmatrix}, \quad C_{2}(\theta)=\begin{pmatrix} A_{12}e^{i(\phi_{12}-\theta)} \\ A_{22}e^{i\phi_{22}} \\ A_{32}e^{i\phi_{32}}\end{pmatrix}.
$$
They are orthogonal to each other and normalized. Defining $C_{3}(\theta) = \epsilon C_{1}(\theta)\wedge C_{2}(\theta)$, we construct  $U(\theta) = (C_{1}(\theta),C_{2}(\theta), C_{3}(\theta))$. This is a family of unitary matrices for which $U(\theta) \to U$ when $\theta \to 0$. By construction, for all $\theta\in\R$, the conditions of Eq.~\eqref{cond:rk4} read:
\[
\left\{\begin{array}{rclc}
\phi_{21}-\phi_{11}  &\neq & \theta & [\frac{\pi}{2}] \\
\phi_{22}-\phi_{12}  &\neq &-\theta &[\frac{\pi}{2}] \\
\phi_{31}-\phi_{11}  & \neq& \theta& [\frac{\pi}{2}] \\
\phi_{32}-\phi_{12}  &\neq &-\theta& [\frac{\pi}{2}] \\\
\phi_{21}-\phi_{11} - \theta &\neq &\phi_{22} - \phi_{12}+\theta& [\pi] \\
\phi_{31}-\phi_{11}-\theta &\neq &\phi_{32}- \phi_{12}+\theta &[\pi] \\
\end{array}\right. 
\]
 {These conditions}   are all fulfilled for $\theta\not=0$ small enough.
This implies that the set $\mathcal W$ is dense.  

To show the set $\mcl{W}$ is of full Haar measure, we show that its complement, $\mcl{W}^{\textrm c}$ is of zero Haar measure. The group $U(3)$ is a $9$-dimensional real manifold. Its Haar measure is absolutely continuous with respect to the Lebesgue measure in any local coordinate patch~\cite{folland2016}. 
Now, $\mcl{W}^{\textrm c}$ is the union of the sets where one of the inequalities in Eq.~\eqref{cond:rk4} is an equality and of the sets where one of the matrix elements of $U$ vanishes. Each of these sets is an lower dimensional submanifold of $U(3)$. Hence it is of zero Lebesgue measure, which concludes the proof.  
\end{proof}
\noindent\emph{Proof of Theorem~\ref{thm:Graal}~(iii).}
We write the entries of a DFT transition matrix $U$ as $U_{kl}=\frac{\omega^{(k-1)(l-1)}} {\sqrt{p}}$ for all $(k,l) \in \IntEnt{1}{p}^2$, where $\omega=e^{-\frac{2i\pi}{p}}$. In this proof, the indices on the matrix $U$ and on all other matrices appearing should be thought of as being extended to all integers and as being periodic with period $p$.

As $\convAB\subseteq \EcalKDC$, we only have to prove that $\EcalKDC \subseteq \convAB$. To that end, we use Corollary~\ref{cor:CSC1}. In other words, we need to show that Eq.\eqref{eq:supcond4} holds for all $\rho\in\EcalKDC$; this is achieved in  Eq.\eqref{eq:DFT2} below. 

We need the following lemma, which characterizes $\VR$ in the case where $U$ is the DFT matrix in prime dimension.
\begin{Lemma}\label{lem:DFTLem} Let $U$ be the DFT matrix in prime dimension $d=p$. Then, a self-adjoint operator $F\in\mathcal{S}_{p}$ belongs to $\VR$ if and only if for all $(i,k)\in\IntEnt{1}{p}^2$,
\begin{equation}\label{eq:DFTF}
    F_{i (i+k)}=F_{(i-k)i}
\end{equation}
Here, $F_{ik}=\langle a_i|F|a_k\rangle$ for $(i,k)\in\IntEnt{1}{p}^2$.
\end{Lemma}
We remark that Eq.\eqref{eq:DFTF} means that the matrix $F$ is constant on its $d-1$ off-diagonals.

\noindent \emph{Proof of Lemma~\ref{lem:DFTLem}.}
 For all $(i,j)\in\IntEnt{1}{p}^2$
\[
Q_{ij}(F) = \sum_{k=1}^{p} \bracket{b_j}{a_i}\bracket{a_k}{b_j}F_{ik} = \frac{1}{p} \sum_{k=1}^{p}\omega^{(j-1)(1-i)}\omega^{(j-1)(k-1)}F_{i k} = \frac{1}{p} \sum_{k=1}^{p}\omega^{(j-1)(k-i)}F_{i k}.
\]
In order to compute $\Im(Q_{ij}(F))$, we rewrite $Q_{ij}(F)$ as follows.
Let $i\in\IntEnt{2}{p}$ and $j\in\IntEnt{2}{p}$. Then, 
\[
Q_{ij}(F) =\frac1{2}\left(Q_{ij}(F)+Q_{ij}(F)\right)=\frac{1}{2p} \left(\sum_{k=1}^{p}\omega^{(j-1)(k-i)}F_{i k} + \sum_{k=1}^{p}\omega^{(j-1)(k-i)}F_{i k}\right).
\]
We now rewrite the second sum. We note that 
$\overline{\omega}^{(j-1)(k'-i)} = \omega^{(j-1)(k-i)}$ if and only if $(j-1)(k-i)=(j-1)(i-k') \ [p]$;  as $(j-1)\neq 0~[p]$, it follows that $k+k'=2i~[p]$ and thus {that} $k'=2i-k~[p]$. As the map $k\in \IntEnt{1}{p} \mapsto (2i-k)~[p]\in\IntEnt{1}{p}$ is bijective, one finds {that}
\[
Q_{ij}(F) = \frac{1}{2p} \left(\sum_{k=1}^{p}\omega^{(j-1)(k-i)}F_{i k}+ \sum_{k=1}^{p}\omega^{(j-1)[(2i-k)-i]}F_{i (2i-k)}\right).
\]
 Note that the indices on $F_{ij}$ are considered modulo $p$. Therefore,
\begin{equation}\label{eq:DFT5}
Q_{ij}(F) =\displaystyle\frac{1}{2p} \left(\sum_{k=1}^{p}\omega^{(j-1)(k-i)}F_{i k}+\overline{\omega}^{(j-1)(k-i)}F_{i (2i-k)}\right).
\end{equation}
By changing the summation index, we {have} 
\[
Q_{ij}(F) =\frac{1}{2p} \left(\sum_{k'=1-i}^{p-i}\omega^{(j-1)k'}F_{i (i+k')}+\overline{\omega}^{(j-1)k'}F_{i (i-k')}\right).
\]
As the indices are considered modulo $p$, the summand is periodic with period $p$, and we can shift the sum to obtain
\[
Q_{ij}(F) =\frac{1}{2p} \left(\sum_{k'=0}^{p-1}\omega^{(j-1)k'}F_{i (i+k')}+\overline{\omega}^{(j-1)k'}F_{i (i-k')}\right).
\]
If $k'\in\IntEnt{1}{\frac{p-1}{2}}$, then $(p-k')\in\IntEnt{\frac{p+1}{2}}{p-1}$ and
\[
\omega^{(j-1)(p-k')}F_{i [i+(p-k')]}+\overline{\omega}^{(j-1)(p-k')}F_{i [i-(p-k')]} = \omega^{(j-1)k'}F_{i (i+k')}+\overline{\omega}^{(j-1)k'}F_{i (i-k')},
\]
so that we can group these terms together. This leads to
\begin{equation}\label{eq:DFTQij}
Q_{ij}(F) =\frac{1}{p} \left(\sum_{k'=1}^{\frac{p-1}{2}}\omega^{(j-1)k'}F_{i (i+k')}+\overline{\omega}^{(j-1)k'}F_{i (i-k')}\right) + \frac{1}{p}F_{ii}.
\end{equation}
We can then finally compute $\text{Im}(Q_{ij}(F))$ for $(i,j)\in\IntEnt{2}{p}^2$:
\[
\begin{array}{rcl}
\text{Im}\left(Q_{ij}(F)\right)  &=&  \displaystyle \frac{1}{p} \sum_{k'=1}^{\frac{p-1}{2}}\text{Im}\left(\overline{\omega}^{(j-1)k'}F_{i (i-k')}+\omega^{(j-1)k'}F_{i (i+k')}\right) \\
&=& \displaystyle \frac{1}{2p\sqrt{-1}} \sum_{k'=1}^{\frac{p-1}{2}}\overline{\omega}^{(j-1)k'}F_{i (i-k')} - \omega^{(j-1)k'}\overline{F_{i (i-k')}} + \omega^{(j-1)k'}F_{i (i+k')} - \overline{\omega}^{(j-1)k'}\overline{F_{i (i+k')}} \\
&=& \displaystyle \frac{1}{2p\sqrt{-1}} \sum_{k'=1}^{\frac{p-1}{2}}\omega^{(j-1)k'}\left(F_{i (i+k')} - \overline{F_{i (i-k')}}\right) + \overline{\omega}^{(j-1)k'}\left(F_{i (i-k')}- \overline{F_{i (i+k')}}\right).
\end{array}
\]
Recall that $F\in\VR$ if and only if, for any $i\in\IntEnt{2}{p}$, the $p-1$ equations $\text{Im}(Q_{ij}(F))=0$ for $j\in\IntEnt{2}{p}$  are satisfied. Indeed, as a consequence of Eq.~\eqref{eq:Qmarginals}, this is equivalent to $\text{Im}(Q_{ij}(F))=0$ for all $(i,j)\in\IntEnt{2}{p}^2$.
Hence $F\in\VR$ if and only if
\[
 \sum_{k'=1}^{\frac{p-1}{2}}\omega^{(j-1)k'}\left(F_{i (i+k')} - \overline{F_{i (i-k')}}\right) + \overline{\omega}^{(j-1)k'}\left(F_{i (i-k')}- \overline{F_{i (i+k')}}\right) = 0.
\]
This system can be rewritten with $z_k= F_{i (i+k)} - \overline{F_{i (i-k)}} \text{ and } z_{p-k} = F_{i (i-k)}- \overline{F_{i (i+k)}}$ for $k\in\IntEnt{1}{\frac{p-1}{2}}$:
\[
A_{\omega}\begin{pmatrix}
z_1 \\
z_2 \\
\vdots\\
z_{p-1}
\end{pmatrix}=0
\]
where 
\[
A_{\omega} = 
\begin{psmallmatrix}
\omega & \omega^2 & \cdots & \omega^{\frac{p-1}{2}} & \overline{\omega}^{\frac{p-1}{2}} & \cdots & \overline{\omega}\\
\omega^2 & \omega^4 &  \cdots & \omega^{p-1} & \overline{\omega}^{p-1} & \cdots &\overline{\omega}^{2} \\
\vdots &\vdots &  \vdots  & \vdots  & \vdots & \vdots  & \vdots  \\
\omega^{p-1} & \omega^{2(p-1)} &  \cdots & \omega^{\frac{(p-1)^2}{2}} & \overline{\omega}^{\frac{(p-1)^2}{2}} & \cdots & \overline{\omega}^{p-1} \\
\end{psmallmatrix}
= 
\begin{psmallmatrix}
\omega & \omega^2 & \cdots & \omega^{\frac{p-1}{2}} &  \omega^{\frac{p+1}{2}} & \cdots & \omega^{p-1} \\
\omega^2 & \omega^4 &  \cdots & \omega^{p-1} &  \omega^{p+1} & \cdots & \omega^{2(p-1)} \\
\vdots &\vdots &  \vdots  & \vdots  & \vdots & \vdots  & \vdots  \\
\omega^{p-1} & \omega^{2(p-1)} &  \cdots & \omega^{\frac{(p-1)^2}{2}} &  \omega^{\frac{(p-1)(p+1)}{2}} & \cdots & \omega^{(p-1)^2} \\
\end{psmallmatrix} .
\]
The matrix $A_{\omega}$ is a Vandermonde matrix, written $V(\omega, \omega^2, \dots, \omega^{p-1})$ for which all parameters are different so $A_{\omega}$ is invertible. This means that for all $k\in\IntEnt{1}{p-1}, z_{k}=0$. Hence, $F\in\VR$ if and only if
\begin{equation}\label{eq:DFT3} \forall(i,k)\in\IntEnt{2}{p}\times\IntEnt{1}{\frac{p-1}{2}}, F_{i (i+k)} = \overline{F_{i (i-k)}}=F_{(i-k)i}.
\end{equation}
 We further rewrite these conditions in a more symmetric form: see Eq.~\eqref{eq:DFT6} below. 
Consider $(i,k)\in\IntEnt{2}{p}\times\IntEnt{\frac{p+1}{2}}{p-1}$. As all  indices are taken modulo $p$,  
$$
F_{i (i+k)} = F_{i (i+k-p)} = F_{i [i-(p-k)]}.
$$
Since $p-k\in\IntEnt{1}{\frac{p-1}{2}}$, Eq.\eqref{eq:DFT3} implies that $F_{i [i-(p-k)]}=\overline{F_{i [i+(p-k)]}} = \overline{F_{i (i-k)}}$. Therefore, we obtain the following recursion relation: 
\begin{equation}\label{eq:DFT4} 
\forall (i,k)\in\IntEnt{2}{p}\times\IntEnt{1}{p-1}, F_{i (i+k)} = \overline{F_{i (i-k)}}=F_{(i-k)i}.
\end{equation}
Next, we want to show that the relation also holds for $i=1$ and $k\in\IntEnt{1}{p-1}$. Suppose $k\in\IntEnt{1}{p-1}$. If $n\in\IntEnt{0}{p-2}$ and since $F$ is self-adjoint,
\[
F_{(nk+1) [(n+1)k+1]} = \overline{F_{[(n+1)k+1](nk+1)}} = \overline{F_{[(n+1)k+1][(n+1)k+1 -k]}}.
\]
As $n+1\neq 0~[p]$ and $k\neq 0~[p]$, it follows that $(n+1)k+1\neq 1~[p]$. We can therefore use Eq.~\eqref{eq:DFT4} to obtain
\[
\forall n\in\IntEnt{0}{p-2}, k\in\IntEnt{1}{p-1},\quad F_{(nk+1) [(n+1)k+1]} = F_{[(n+1)k+1][(n+1)k+1+k]} = F_{[(n+1)k+1][(n+2)k+1]}.
\]
It follows from this that for all $n\in\IntEnt{1}{p-1}$,
$F_{1(k+1)} = F_{(nk+1)[(n+1)k+1]}$. And thus, $F_{1(1+k)} = F_{(1-k)1} $ which is the above relation  for $n=p-1$. Thus, this shows that Eq.\eqref{eq:DFT3} holds for $i=1$.

Summing up, $F\in\EcalKDC$ if and only if 
\begin{equation}\label{eq:DFT6}
\begin{array}{rcl}
F_{i (i+k)}&=&F_{(i-k)i}\text{ for all }(i,k)\in\IntEnt{1}{p}\times\IntEnt{1}{p}.
\end{array}
\end{equation}
\qed
\\

We can now use this result to show that $\rho\in\EcalKDC$ implies that $\rho\in\convAB$, by showing that  Eq.~\eqref{eq:supcond4} holds. Indeed, since $\rho\in\EcalKDC$ implies that $\rho\in\VR$, it follows from Eq.~\eqref{eq:DFTQij}  and Lemma~\ref{lem:DFTLem} that for all $(i,j)\in\IntEnt{1}{p}^2$, 
\begin{equation}\label{eq:DFT2}
\begin{array}{rcl}
Q_{ij}(\rho)+Q_{11}(\rho) &=&  \displaystyle \frac{2}{p} \sum_{k=1}^{\frac{p-1}{2}}\text{Re}\left(\omega^{(j-1)k'}\rho_{i (i+k')}\right) + \frac{1}{p}\rho_{i i}+  \frac{2}{p} \sum_{k=1}^{\frac{p-1}{2}}\text{Re}\left(\rho_{1 (1+k')}\right) + \frac{1}{p}\rho_{1 1} \\
&=& \displaystyle \frac{2}{p} \sum_{k=1}^{\frac{p-1}{2}}\text{Re}\left(\omega^{(j-1)k'}\rho_{1 (1+k')}\right) + \frac{1}{p}\rho_{1 1}+  \frac{2}{p} \sum_{k=1}^{\frac{p-1}{2}}\text{Re}\left(\rho_{i (i+k')}\right) + \frac{1}{p}\rho_{i i} \\
&=& Q_{1j}(\rho)+ Q_{i1}(\rho).
\end{array}
\end{equation}
This establishes the relations~\eqref{eq:supcond4}  for $(i,j)\in\IntEnt{1}{p}^2$ and $k=l=1$. As in the proof of Proposition~\ref{prop:CSC2}, this implies that they  hold for all $k,l\in\IntEnt{1}{p}^2$. Thus, we have proven that $\EcalKDC \subseteq \convAB$. This ends the proof.

\noindent{\bf{Remark :}} An alternative proof of Theorem \ref{thm:Graal}.(iii) can be obtained as follows. Lemma \ref{lem:DFTLem} implies that $F$ is constant on its $(d-1)$ off-diagonals and as {$F$} is self-adjoint, only $\frac{d-1}{2}$ of these values are independent. Hence, the off-diagonals of F are determined by $(d-1)$ real parameters. The diagonal of {$F$} contains $d$ real parameters. Lemma \ref{lem:DFTLem} implies that $\VR$ is a $(2d-1)$ real vector space. Proposition \ref{prop:TSKD2bis}$(ib)$ then implies that $\EcalKDC =  \convAB$.
\vskip0.2cm

\noindent{\it Proof of Theorem~\ref{thm:Graal}~(iv).} 
Note that this statement means that the set of $U$ for which Eq.~\eqref{eq:Graalagain} holds is open. The result follows from the following Proposition.
\begin{Prop}\label{prop:stability1} Let $U$ be such that $\mab>0$ and $\EcalKDC=\convAB$. Let $U_\epsilon$ be a family of unitary transition matrices between bases $\Acal_\epsilon$ and $\Bcal_\epsilon$ satisfying $\lim_{\epsilon\to 0}U_\epsilon=U$. Then, for all $\epsilon$ sufficiently small, one has
$\EcalKDC^\epsilon=\conv{\Acal_\epsilon,\Bcal_\epsilon}$. 
\end{Prop}
Proposition \ref{prop:stability1} states that the set of $U$ for which  $\EcalKDC=\convAB$ is an open set, so this proposition proves part (iv) of Theorem~\ref{thm:Graal}.
\begin{proof} Consider 
$$
\mathcal L=\{M\in\mathcal{M}_d(\R)\mid \sum_i M_{ij}=0,\ \sum_j M_{ij}=0\},
$$
which is  a $(d-1)^2$-dimensional real vector space. As a result of Eq.~\eqref{eq:Qmarginals}, one has
$$
\mathrm{Im}Q_\epsilon :\SAO\to \mathcal L.
$$
Here, $Q_\epsilon(\, \cdot \, )$ is the KD distribution associated to $U_\epsilon$. Suppose that, for $\epsilon=0$, $\EcalKDC=\convAB$. 
Then, according to Proposition~\ref{prop:TSKD2bis}, $\dim(\Ker{(\mathrm{Im} Q)}  )=2d-1$ and hence, since $\dim \SAO=d^2$, it follows that $\Im Q$ is surjective. We now show that, for sufficiently small $\epsilon$, $\Im Q_{\epsilon}$ is also surjective. Writing 
$$
\SAO= \Ker(\mathrm{Im} Q)\oplus [\Ker(\mathrm{Im} Q)]^\perp,
$$
it follows that $\widehat{\mathrm{Im} Q}:= \mathrm{Im}Q_{\mid [\Ker(\mathrm{Im} Q)]^\perp}$  is a linear isomorphism between $ [\Ker(\mathrm{Im} Q)]^\perp$ and $\mathcal L$. It therefore has an inverse $\widehat{\mathrm{Im} Q}^{-1}$.  Let us now write 
$$
\delta U_\epsilon:=U_\epsilon -U,
$$
with $\delta U_\epsilon\to0$ as $\epsilon\to0$. Then, 
$$
\widehat{\delta \mathrm{Im} Q_\epsilon}:= \widehat{\mathrm{Im}Q_\epsilon}-\widehat{\mathrm{Im}Q},
$$
with $\delta\mathrm{Im} Q_\epsilon\to0$. 
We now consider $\widehat {\mathrm{Im}Q_\epsilon}:=\mathrm{Im}Q_{\epsilon\mid [\Ker(\mathrm{Im} Q)]^\perp}$, so that
$$
\widehat {\mathrm{Im}Q_\epsilon} =\widehat  {\mathrm{Im} Q}+\delta \widehat{\mathrm{Im}Q_\epsilon} : [\Ker(\mathrm{Im} Q)]^\perp\to \mathcal L.
$$ 
 To conclude the proof,  we show  that $\widehat {\mathrm{Im}Q_\epsilon}$ is a linear isomorphism. One has
$$
\widehat {\mathrm{Im}Q_\epsilon}=\widehat {\mathrm{Im}Q}(\widehat \bbone +\widehat{\mathrm{Im}Q}^{-1}\delta\widehat {\mathrm{Im}Q_\epsilon}).
$$
Since $(\widehat \bbone +\widehat{\mathrm{Im}Q}^{-1}\delta\widehat {\mathrm{Im}Q_\epsilon})$ is a small perturbation of the identity, it is invertible. So $\widehat {\mathrm{Im}Q_\epsilon}$ is invertible as the composition of two invertible maps. This implies that $ \mathrm{Im} Q_\epsilon$ is surjective and, hence, that the $\dim\left(\Ker({\Im{Q_{\epsilon}}})\right) = 2d-1$. Proposition \ref{prop:TSKD2bis} then implies that $\EcalKDC^{\epsilon} = \conv{\Acal_{\epsilon},\Bcal_{\epsilon}}$.
\end{proof}

\noindent{\bf Remark.}
We note that, in dimension $d=2$, if there is a zero in $U$, then $U$ is, up to phase changes, either equal to $\begin{pmatrix} 1 & 0 \\ 0 &1 \end{pmatrix}$ or to $\begin{pmatrix} 0 & 1 \\ 1 &0 \end{pmatrix}$. In that case, the two bases are not distinct and all pure states are KD-positive, and thus all mixed state are also KD-positive. In higher dimension $d\geq 3$, the presence of zeroes in $U$ considerably complicates the analysis. \\

We  conjecture that Theorem~\ref{thm:Graal}~(ii) is true in all dimensions $d\geq 2$. We numerically checked this conjecture in dimensions $d$ up to $10$.  For that purpose, we sampled random unitary matrices according to the Haar measure on the unitary group and computed numerically the rank of $\Im Q$ for each such matrix. When it equals $(d-1)^2$, Proposition~\ref{prop:TSKD2bis}~(ib) guarantees that $\EcalKDC=\convAB$. We did not find any instance where this condition was not satisfied. 

We now prove Proposition~\ref{prop:pure_char}, which can be seen as a first step in the proof of this conjecture. For convenience, we repeat it here:

\begin{propIntro}
Let $d\geq2$.
There exists an open dense set of unitaries of probability $1$ for which $\EcalKDCpu=\Acal\cup\Bcal$.
\end{propIntro}
\begin{proof}
We define $\mcl{V}$ to be the set of $d$ by $d$ unitary matrices $U$ satisfying $\mab > 0$ and 
\begin{equation}\label{eq:proof_pure_char}
\forall   (k,k',j,j')\in \IntEnt{1}{d}^4, k\neq k', j\neq j', \  \phi_{k,j} -  \phi_{k',j} \neq  \phi_{k,j'} -  \phi_{k',j'} [2\pi].
\end{equation}
Here we wrote, as before, $U_{kj}=\langle a_k|b_j\rangle=A_{kj}\exp(i\phi_{kj})$, with $A_{kj}>0$. 
We will first show that, if $U\in\mcl{V}$, then the only pure KD-positive states are the basis states. We proceed by contradiction.  Suppose that there exists  a pure KD-classical state $\cket{\psi}\in\H$ that is not a basis state. By reordering the two bases, we can suppose that $S_{\Acal}(\psi)=\IntEnt{1}{n_{\Acal}(\psi)}$ and $S_{\Bcal}(\psi)=\IntEnt{1}{n_{\Bcal}(\psi)}$ with $n_{\Acal}(\psi) \geqslant 2$ and $n_{\Bcal}(\psi) \geqslant 2$.
Then, we change the phases of the basis states as follows: $\cket{a_j}$ is changed to $e^{-i\alpha_j}\cket{a_j}$ for $j\in\IntEnt{1}{n_{\Acal}(\psi)}$ where $\alpha_j=\arg(\bracket{a_j}{\psi})$ and  $\cket{b_j}$ is changed to $e^{i\beta_j}\cket{b_j}$ for $j\in\IntEnt{1}{n_{\Bcal}(\psi)}$ where $\beta_j=\arg(\bracket{b_j}{\psi})$. Thus, for all $(i,j)\in\IntEnt{1}{n_{\Acal}(\psi)}\times\IntEnt{1}{n_{\Bcal}(\psi)}$,
\[
 \bracket{a_i}{\psi}\in\R^{+}_*, \bracket{\psi}{b_j}\in\R^{+}_* \ \mathrm{ and } \ Q_{ij}(\psi)=A_{ij}e^{i(\phi_{ij}+\alpha_i+\beta_j)} \bracket{a_i}{\psi}\bracket{\psi}{b_j}\in\R^{+}_*.
\]
Consequently, 
\[
\forall (i,j)\in\IntEnt{1}{n_{\Acal}(\psi)}\times\IntEnt{1}{n_{\Bcal}(\psi)}, \phi_{ij}+\alpha_i+\beta_j = 0~[2\pi].
\]
As $n_{\Acal}(\psi) \geqslant 2$ and $n_{\Bcal}(\psi) \geqslant 2$, it follows that
\[
\phi_{11} - \phi_{12} = \beta_2 - \beta_1 ~[2\pi] \ \mathrm{ and } \ \phi_{21} - \phi_{22} = \beta_2 - \beta_1 ~[2\pi].
\]
Thus,
\[
\phi_{11} - \phi_{12} =\phi_{21} - \phi_{22} ~[2\pi],
\]
which is a contradiction because $U\in\mcl{V}$. Therefore, the only KD-classical pure states associated to $U\in\mcl{V}$ are the basis states.

We now show that the set $\mcl{V}$ is an open and dense set. That it is open follows directly from its definition. It remains to show that it is dense. For that purpose, we will show below that the set $\mcl{Z}$ defined by 
\[
\phi_{11} - \phi_{12} \neq \phi_{21} - \phi_{22} ~[2\pi]
\ \mathrm{and} \ \mab > 0
\]
is dense. Reordering the basis elements, it then follows that all sets defined by 
$$
\phi_{k,j} -  \phi_{k',j} \neq  \phi_{k,j'} -  \phi_{k',j'} ~[2\pi]
$$
are dense. One concludes that $\mathcal V$ is dense as a finite intersection of dense sets. 
It remains to prove that $\mathcal Z$ is dense. Let $U\in\mcl{Z}$; hence $\mab > 0$ and $\phi_{11} - \phi_{12} = \phi_{21} - \phi_{22} ~[2\pi]$. We denote by $(C_i)_{i\in\IntEnt{1}{d}}$ the columns of $U$. We define, for $\theta\in \R$
\begin{equation}\label{eq:pure_char_1}
C_{1}(\theta) = \begin{pmatrix}
A_{11}e^{i(\phi_{11}+\theta)} \\
A_{21}e^{i\phi_{21}}\\
\vdots\\
A_{d1}e^{i\phi_{d1}}
\end{pmatrix} \ \mathrm{ and } \ C_{2}(\theta) = \begin{pmatrix}
A_{12}e^{i(\phi_{12}-\theta)} \\
A_{22}e^{i\phi_{22}}\\
\vdots\\
A_{d2}e^{i\phi_{d2}}
\end{pmatrix}.
\end{equation}
Note that $C_1(\theta)$ and $C_2(\theta)$ are normalized and orthogonal. By applying the Gram-Schmidt algorithm to $\begin{pmatrix}
    C_{1}(\theta) & C_{2}(\theta) & C_{3} & \cdots & C_{d}
\end{pmatrix}$, we obtain a unitary matrix 
$$
U_{\theta}=\begin{pmatrix}
    C_{1}(\theta) & C_{2}(\theta) & C_{3}(\theta) & \cdots & C_{d}(\theta)
\end{pmatrix}
$$
such that $U(\theta) \rightarrow U$ when $\theta \rightarrow 0$. 
Therefore, for $\theta\not=0$ small enough, $\mab(\theta) > 0$ and 
\[
\phi_{21} - \phi_{11} - \theta \neq \phi_{22} - \phi_{12}+\theta \ [2\pi].
\]
This proves $\mcl{V}$ is dense. 

We finally show  that $\mcl{V}$ is a set of probability one for the unique normalized Haar measure on the unitary group.  Note that $\mcl{V}$ is contained in the complement of the union of the subsets  
$$
\phi_{k,j} -  \phi_{k',j} =  \phi_{k,j'} -  \phi_{k',j'} \ [2\pi]
$$
and of the subsets where one of the elements of $U$ vanishes.  Each of those subsets is a lower dimensional submanifold of the unitary group.  Since the Haar measure is known to be absolutely continuous with respect to Lebesgue measure in any local coordinate system on the unitary group~\cite{folland2016}
, this implies {that} the Haar measure of these manifolds vanishes. The same property therefore holds for their union, so that indeed $\mcl{V}$ has measure $1$.

\end{proof}

Proving Eq.~\eqref{eq:Graalagain} for a particular $U$ can be hard, as the result on the DFT in prime dimensions (Theorem~\ref{thm:Graal}~(iii)) shows. It is certainly not always true. To see this, one may first note that for the DFT in non-prime dimensions, it is well known (see for example~\cite{ debievre2023a,Xu22}) that
$
\Acal\cup\Bcal\subsetneq\EcalKDCpu.
$
We do not know, however, if in {this} case 
$
\EcalKDCpu\subsetneq \EcalKDCext
$
. In the next section we construct examples where
$$
\Acal\cup\Bcal\subsetneq \EcalKDCpu\subsetneq \EcalKDCext.
$$

We further point out again that Eq.~\eqref{eq:Graalagain} is notably different from what happens for the continuous-variable Wigner positivity. Indeed, there  exist Wigner positive states outside the convex hull of the pure Wigner positive states \cite{hudson1974, takagizhuang2018}. We note also that the discrete-variable Wigner function in $d=3$ has positive states outside the convex hull of its pure positive states~\cite{gross2006}. This is not the case for the KD distribution associated to the DFT, as a result of Theorem~\ref{thm:Graal}~(iii).

\section{Extreme KD-positive states are not necessarily pure}\label{s:eurekanew}

Below, we construct examples where
$$
\Acal\cup\Bcal\subsetneq \EcalKDCpu\subsetneq \EcalKDCext\quad \textrm{or}\quad \Acal\cup\Bcal = \EcalKDCpu\subsetneq \EcalKDCext.
$$
In other words, in these cases there exist mixed extreme KD-positive states. In particular, when such states are viewed as a convex combination of pure states, at least one of those pure states must be KD-negative. In fact, the convex combination of any state in
$$
\EcalKDC\setminus \conv{\EcalKDCpu},
$$
must include pure KD-negative states.
Proposition~\ref{prop:TSKD2bis} states that $\EcalKDCext=\Acal\cup\Bcal$ is true if and only if the space $\VR$ of KD-real operators is of its smallest possible dimension: $\dim \VR = 2d-1$ (see Eq.~\eqref{eq:dimVRbounds}). In Section \ref{subsec:realTran}, we show  that this is never the case if $U$ is a real (hence orthogonal) matrix and $d\geq 3$ (Lemma~\ref{lem:realU}).

Let us point out that this statement does not contradict Theorem~\ref{thm:Graal}~(ii), which states that in dimension $d=3$, the equality $\EcalKDC=\convAB$ holds with probability one; nor does it contradict our conjecture that it holds for all $d\geq 3$. Indeed, the space of unitary matrices is of dimension $d^2$; the space of orthogonal matrices is only of dimension $\frac{d(d+1)}{2}$, i.e., a ``thin'' subset. More formally, the space of orthogonal matrices is a hypersurface of lower dimension with empty interior and is therefore of zero probability among all unitary matrices.

 The result  of Lemma~\ref{lem:realU} allows us to construct examples  of bases $\Acal$ and $\Bcal$ for which $\EcalKDCpu\subsetneq\EcalKDCext$. 
In Section~\ref{s:mixedextremed=3}, we  provide an explicit example of an orthogonal matrix $U^\star$ in $d=3$ for which
$$
\Acal\cup\Bcal\subsetneq \EcalKDCpu\subsetneq \EcalKDCext.
$$
In Section~\ref{s:spin1}, we first show that the transition matrix between the bases of two different spin components of a spin-$s$ system is always (equivalent to) a real transition matrix. 
We then show that the a matrix $U^\star$ constructed in Secion~\ref{s:mixedextremed=3}  arises for a spin-$1$ system, in which the two bases $\Acal$ and $\Bcal$ correspond to the eigenvectors of the spin component in the $z$ direction and another, specific direction, respectively.

Finally, in Section~\ref{s:mixedextremed=4n}, we show that there exist examples in all dimensions $d=2^n$ (for integer $n$) and in all dimensions $d=4m$ (for integer $m \leq 166$)  for which 
$$
\Acal\cup\Bcal = \EcalKDCpu\subsetneq \EcalKDCext.
$$
In these cases, the two bases are perturbations of real MUB bases.  

For completeness, we mention that we did not find any example where
$$
\Acal\cup\Bcal\subsetneq \EcalKDCpu= \EcalKDCext.
$$

\subsection{Real transition matrices: structural results}\label{subsec:realTran}
\begin{Lemma}\label{lem:realU}
If $U$ is real and $\mab >0$, then $\VR={\SAOr}$, where ${\SAOr}$ is the set of self-adjoint operators that have a real and symmetric matrix on the $\Acal$ (and hence also on the $\Bcal$) basis. Hence, $\dim\VR=\frac{d(d+1)}{2}$. If in addition $d\geq 3$, then $\Acal\cup\Bcal\subsetneq \EcalKDCext$.
\end{Lemma}
Note that $\frac{d(d+1)}{2}$ is strictly larger than $2d-1$ for all $d>2$. We can further interpret this lemma as follows.   If $U$ is real, then the observables $F$ that have a real KD symbol are precisely those described by a real symmetric matrix on the $\Acal$ and $\Bcal$ bases. This constitutes a concrete identification of $\VR$, which is not available in general. 
In this situation, the kernel of $\textrm{Im} Q$, which is $\VR$, is large and in particular, larger than $\spanRAB$. 
\begin{proof} Suppose that $F\in\VR$, so
\[
\forall (i,j)\in\IntEnt{1}{d}^2, \; \Im(Q_{i j}(F))=0.
\]
We  know also that $\forall (i,j)\in\IntEnt{1}{d}^2, \; Q_{i j}(F) = \bracket{b_j}{a_i}\bra{a_i}F\cket{b_j}= \bracket{b_j}{a_i}\sum_{k=1}^{d}\bracket{a_k}{b_j}\bra{a_i}F\cket{a_k}$.
As $\bracket{b_j}{a_i}\bracket{a_k}{b_j}\in\R$, we finally find that
\[
\Im(Q_{i j}(F))=\sum_{k=1, k\neq i}^{d}\bracket{b_j}{a_i}\bracket{a_k}{b_j}\Im{\left(\bra{a_i}F\cket{a_k}\right)} = \bracket{b_j}{a_i} \left( \sum_{k=1, k\neq i}^{d}\Im{\left(\bra{a_i}F\cket{a_k}\right)} \bracket{a_k}{b_j}\right) =0. 
\]
Since this equation holds for all $j$, by fixing $i\in\IntEnt{1}{d}$ we can write
\[
\forall j\in\IntEnt{1}{d}, \sum_{k=1, k\neq i}^{d}\Im{\left(\bra{a_i}F\cket{a_k}\right)} \bracket{a_k}{b_j}=0
\] 
and thus
\[
\sum_{k=1, k\neq i}^{d}\Im{\left(\bra{a_i}F\cket{a_k}\right)}\cket{a_k}=0.
\]
So, for all $k\in\IntEnt{1}{d}, \Im{\left(\bra{a_i}F\cket{a_k}\right)}=0$ and thus, $F\in \SAOr$.

For the last statement, note that, if $\EcalKDCext=\Acal\cup\Bcal$, then $\EcalKDC=\convAB$ and hence, by Proposition~\ref{prop:TSKD2bis}, $\dim\VR=2d-1$, which is a contradiction with the fact that dim$\VR=\frac{d(d+1)}{2}$.
\end{proof}

As announced above, it is the goal of this section to exhibit mixed KD-positive states that cannot be written as convex combinations of the pure KD-positive states. In order to find and analyse such states, we will concentrate on dimension $d=3$ where the analysis is tractable. 

Note that, when $d=3$, then  $\dim \VR=6$ and $\dim (\spanRAB)=5$, so that
$\spanRAB$ is a co-dimension $1$ subspace of $\VR$. We characterize the one-dimensional subspace of $\VR$ perpendicular to $\spanRAB$, as follows. We use the Hilbert-Schmidt inner product on $\SAOr$. Then, $F_\perp\in \SAOr$ is a unit vector orthogonal to $\spanRAB$ if and only if
$$
\forall i\in\IntEnt{1}{3}, \ \Tr(F_\perp\cket{a_i}\bra{a_i})=0, \ \Tr(F_\perp\cket{b_i}\bra{b_i})=0, \ \mathrm{ and } \  \Tr F^2_\perp=1.
$$
It follows that the matrix of $F_\perp$ in the $\Acal$-basis is of the form
\begin{equation}\label{def:F_perp}
F_\perp=\frac1{\sqrt2}
\begin{pmatrix}
0& f_1& f_2\\
f_1&0&f_3\\
f_2&f_3&0
\end{pmatrix}
\end{equation}
with $f=(f_1, f_2, f_3)\in\R^3$ and 
$$
f_1^2+f_2^2+f_3^2=1.
$$
The vector $f$ is, up to a sign, uniquely determined by the conditions $\Tr(F_\perp|b_j\rangle\langle b_j|)=0$ for all $j\in\IntEnt{1}{3}$. 
Since $\Tr F_\perp=0$, $F_\perp$ can  be neither positive nor negative. We can assume that $f_-<0 \leq f_0\leq f_+$ are its three eigenvalues.
Consequently,
\begin{equation}
\VR=\spanRAB\oplus \R F_\perp.
\end{equation}
Any $F\in\VR$ can then be decomposed uniquely as $F=F_{\mathrm{b}}+xF_\perp$, with $F_{\mathrm{b}}\in\spanRAB$ and $x\in\R$. Note that
\begin{equation}\label{eq:purity}
\Tr F=\Tr F_{\mathrm{b}}, \quad \Tr F^2=\Tr F_{\mathrm{b}}^2 +x^2.
\end{equation}

We finally note that it follows from results in~\cite{debievre2023a}  that, when $d=3$ and $\mab>0$, there is a finite set of pure KD-positive  states. Further structural information on the set $\EcalKDC$ and its extreme points, for real orthogonal $U$, is given in Lemma~\ref{lem:EKDCorthV2} and Proposition~\ref{prop:EKD2}. 

The above results can be summarized as follows. 
There exists a convex subset $\mathcal D$ of $\spanRAB$, real numbers $x_{\min} \leq 0\leq  x_{\max}$, a concave function 
$$
x_{+}:\Dcal\to [x_{\min}, x_{\max}]\subset \R,
$$
as well as  a  convex function
$$
x_-:\mathcal D\to [x_{\min}, x_{\max}]\subset \R,
$$
so that 
$$
\EcalKDC=\{\sigma +[x_-(\sigma), x_+(\sigma)]F_\perp\mid \sigma\in\Dcal \}.
$$
In other words, $\EcalKDC$ can be seen as the intersection of the subgraph of $x_+$ and of the supergraph of $x_-$.
As a result, the extreme points of $\EcalKDC$ lie on the graphs of $x_-$ and of $x_+$:
$$
\EcalKDCext\subseteq x_+(\mathcal D)\cup x_-(\mathcal D).
$$
Since the precise form of the functions $x_{\pm}$ as well as of the set $\mathcal D$ depend in a nontrivial manner on the matrix $U$, determining explicitly the nature of the set of extreme KD-positive states is far from straigthforward for general real $U$, even in dimension $3$. In the next subsection, we study a special example where the nature of the extreme KD-positive states can be mapped out in more detail.  In particular, we show that $\EcalKDC$ is not necessarily a polytope.

\subsection{Extreme KD-positive  states can be mixed: an example in $d=3$.}\label{s:mixedextremed=3}

We now provide an example of an orthogonal $U$ in $d=3$ for which
$$
\Acal\cup\Bcal\subsetneq \EcalKDCpu\subsetneq \EcalKDCext.
$$
The first strict inclusion follows from Proposition~\ref{prop:TSKD2bis} and says  that there exist additional pure KD-positive  states distinct from the basis states. In our example, there will, in addition, be mixed extreme KD-positive  states. In other words, the polytope $\conv{\EcalKDCpu}$ does not exhaust all of $\EcalKDC$. 

To motivate our choice of $U$, we first recall that it was shown in~\cite{debievre2023a} (Theorem~13) that, when $d=3$, any $U$ for which $\frac{\mab}{\Mab}>\frac{1}{\sqrt{2}}$ has the property that $\EcalKDCpu=\Acal\cup\Bcal$. If there existed such $U$ among the real orthogonal matrices, this would imply by Lemma~\ref{lem:realU} that $\EcalKDCpu\subsetneq \EcalKDCext$. However, such real orthogonal matrices do not exist. The largest value of $\frac{\mab}{\Mab}$ that can be obtained for real orthogonal matrices is $\frac{1}{2}$, attained for the matrix
\begin{equation}\label{eq:specialU}
U^\star=\frac13
\begin{pmatrix} 
-1&2&2\\
2&-1&2\\
2&2&-1
\end{pmatrix}.
\end{equation}
 We further know from~\cite{debievre2021} that in $d=3$, a pure KD-positive  state $\cket{\psi}$ of any transition matrix with $\mab>0$ satisfies $\nab(\psi) = n_{\mathcal{A}}(\psi) + n_{\mathcal{B}}(\psi) =4$, where 
\begin{equation}\label{eq:nab}
\na(\psi)=\sharp\{i\in\llbracket 1,d\rrbracket \mid \langle a_i|\psi\rangle\not=0\}, \ \mathrm{ and } \ \nb(\psi)=\sharp\{j\in\llbracket 1,d\rrbracket \mid \langle b_j|\psi\rangle\not=0\}.
\end{equation}
Here, $\sharp \{ \cdot \}$ denotes the cardinality of $ \{ \cdot \}$.
One can  construct all such states and check if they are KD-positive. Doing so, we find, in addition to the basis states, the following three KD-positive  states:
\[
\cket{\varphi_3} = \frac{\cket{a_1}-\cket{a_2}}{\sqrt{2}},\
\cket{\varphi_2} = \frac{\cket{a_3}-\cket{a_1}}{\sqrt{2}}, \ \mathrm{ and } \
\cket{\varphi_1} = \frac{\cket{a_2}-\cket{a_3}}{\sqrt{2}}.
\]
Their associated KD distribution\blue{s} are
\[
Q(\varphi_3) = \frac{1}{6}\begin{pmatrix} 1 & 2 & 0 \\ 2 & 1 & 0 \\ 0 & 0 & 0 \end{pmatrix}, \ Q(\varphi_2) = \frac{1}{6}\begin{pmatrix} 1 & 0 & 2 \\ 0 & 0 & 0 \\ 2 & 0 & 1 \end{pmatrix}, \ \mathrm{ and } \ Q(\varphi_1) = \frac{1}{6}\begin{pmatrix} 0 & 0 & 0 \\ 0 & 1 & 2 \\ 0 & 2 & 1 \end{pmatrix} ,
\]
respectively. Here, to lighten the notation, we write $Q(\varphi_i):=Q(|\varphi_i\rangle\langle \varphi_i|)$.
The operator $F_\perp$ in Eq.~\eqref{def:F_perp} is readily computed to be 
$$
F_\perp=\frac1{\sqrt6}
\begin{pmatrix}
0&1&1\\
1&0&1\\
1&1&0
\end{pmatrix}.
$$
A simple computation shows that 
$$
|\varphi_i\rangle\langle \varphi_i|=F_{b}^{i} +x_iF_\perp,
$$
with $x_i=\langle \varphi_i| F_\perp |\varphi_i\rangle=-\frac{1}{\sqrt{6}}$ and  
\[
F_{b}^{i}= \frac{5}{6}\sum_{k\neq i} \cket{b_k}\bra{b_k} + \frac{1}{12}\cket{b_i}\bra{b_i} -\frac{3}{4}\cket{a_i}\bra{a_i}.
\]
Hence, that the triangle  with vertices $|\varphi_i\rangle\langle \varphi_i|$ lies in a  plane parallel to $\spanRAB$ at a distance $|x_i|$ from it.  It follows that $\conv{\EcalKDCpu}$ lies between $x=-\frac{1}{\sqrt6}$ and $x=0$. Indeed, if 
\[
\rho = \sum_{i=1}^{3} \lambda_{i}\cket{a_i}\bra{a_i} + \mu_{i}\cket{b_i}\bra{b_i} + \delta_{i}\cket{\varphi_i}\bra{\varphi_i}
\]
with  $\sum_{i=1}^{3} \lambda_{i}+ \mu_{i}+ \delta_{i}=1$ where $\lambda_{i}\geqslant 0, \mu_{i}\geqslant 0, \ \mathrm{ and } \ \delta_{i}\geqslant 0$ for all $i\IntEnt{1}{3}$, then 
\[
x_{\rho} = \mathrm{Tr}(\rho F_\perp) = \sum_{i=1}^{3} \delta_{i}\bra{\varphi_i}F_\perp\cket{\varphi_i} = -\frac{1}{\sqrt{6}}\sum_{i=1}^{3} \delta_{i} .
\]
Thus,  $-\frac{1}{\sqrt{6}}\leqslant x_{\rho} \leqslant 0$. It follows that any  KD-positive  states for which $x>0$ cannot belong to $\conv{\EcalKDCpu}$.
We now construct such states. 

We consider $\rho({x})=\frac{1}{3}I_{d} + xF_\perp$. Its  eigenvalues are $\left\{\frac{2(1+x\sqrt{6})}{6},\frac{2-x\sqrt{6}}{6}, \frac{2-x\sqrt{6}}{6} \right\}$ which are positive if and only if $x\in[-\frac{1}{\sqrt{6}},\frac{2}{\sqrt{6}}]$. Only in these cases is $\rho(x)$ a quantum state. 
Moreover,
\[
Q(F_\perp) =\frac{\sqrt{6}}{27} \begin{pmatrix}
-2 & 1 & 1 \\
1 & -2 & 1 \\
1 & 1 & -2
\end{pmatrix} ,
\]
and consequently 
\[
Q(\rho(x)) = \frac{1}{27}
\begin{pmatrix}
1-2x\sqrt{6} & 4+x\sqrt{6} & 4+x\sqrt{6} \\
4+x\sqrt{6} & 1-2x\sqrt{6}& 4+x\sqrt{6}\\
4+x\sqrt{6} & 4+x\sqrt{6} & 1-2x\sqrt{6}
\end{pmatrix}.
\]
So $\rho(x)$ is KD-positive if and only if $x\in [-\frac{4}{\sqrt{6}},\frac{1}{2\sqrt{6}}]$. Finally, $\rho(x)$ is a KD-positive  state provided that $x\in[-\frac{1}{\sqrt{6}},\frac{1}{2\sqrt{6}}]$.

Thus, for all $x\in ] 0,\frac{2}{3}]$, $\rho(x)$ is a KD-positive  mixed state and from what precedes, it follows that $\rho(x)$ is not in the convex hull of the KD-positive  pure states. This construction therefore exhibits explicit KD-positive  mixed states that are not in the convex hull of $\EcalKDCpu$.  In addition, Lemma~\ref{lem:EKDCorthV2} allows the generalisation of   this construction, which shows that, for every $\rho\in\mathrm{Int}(\convAB)$,  there exists a continuous family of states $\rho(x) = \rho+x F_{\perp}$ which are not in the convex hull of KD-positive  pure states.  Here and elsewhere, $\mathrm{Int}(A)$ stands for the interior of $A$. 

Since $\EcalKDC$ is a convex and compact set it follows from the Krein-Milman Theorem \cite{hiriart-urrutylemarechal2001} that  it is the convex hull of its extreme points. We thus reach the conclusion that $\EcalKDC$ has extreme points that are not pure and not in the convex hull of the pure KD-positive  states.  We identify some of them explicitly in Appendix~\ref{subs:KDCextmixed}. It is known that mixed extreme states   exist also for the discrete variable Wigner function (at least in $d=3$~\cite{gross2006}), as well as in the continuous variable Wigner function~\cite{genonietal2013,vanherstraetencerf2021, hertzdebievre2023}. However, to the best of our knowledge, such states have never been explicitly identified.

\noindent{\bf Comment.} We point out that the set $\EcalKDC$ is in the example above  not a polytope, since it has an infinite number of extreme points. To see this, note that,  if $\EcalKDC$ was a polytope, any 2-dimensional section of it would be a polygon. However, if we consider the 2-dimensional plane containing $F_{\perp}$, $\frac{1}{2}(\cket{a_1}\bra{a_1} + \cket{a_2}\bra{a_2})$ and $\cket{a_3}\bra{a_3}$ and we intersect it with $\EcalKDC$, then simple computations show that this intersection is not a polygon, as illustrated in Fig.\ref{fig:polytope}. Hence, $\EcalKDC$ is not a polytope.

\begin{figure}
    \centering
    \includegraphics[scale=0.5]{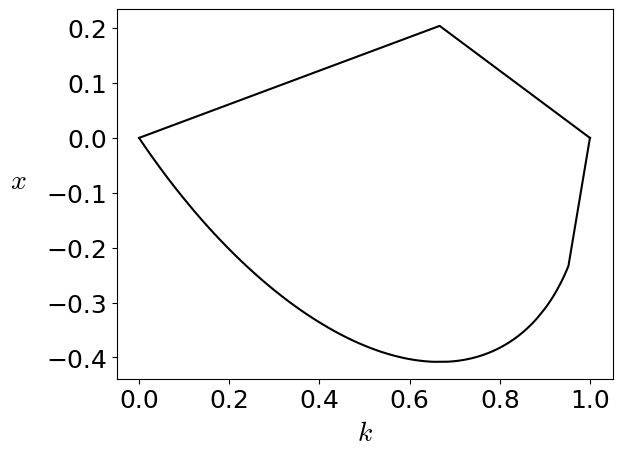}
    \caption{A point $(k,x)$ on this graph represents the state $\frac{k}{2}(\cket{a_1}\bra{a_1}+\cket{a_2}\bra{a_2}) + (1-k)\cket{a_3}\bra{a_3} + x F_{\perp}$. A state is KD positive if and only if it lies inside the region drawn. Thus, this figure displays a 2-dimensional section of $\EcalKDC$ that is not a polygon. This shows that $\EcalKDC$ is not a polytope.}
    \label{fig:polytope}
\end{figure}
\subsection{An application: the case of spin-$s$}\label{s:spin1}
In this subsection, we show how the transition matrix $U^\star$ of the previous subsection arises naturally in a spin-$1$ system. 

First, we show that the transition matrices between spin component bases for spin-$s$ systems are (equivalent to) a real matrix.  Let 
$|z,m\rangle$
be the standard basis vectors of $J_z$, with eigenvalues $m=-s, \dots, s$. Let 
$$
R(\alpha,\beta,\gamma)=R_z(\alpha)R_y(\beta)R_z(\gamma) ,
$$
be the rotation matrix with Euler angles $\alpha,\beta,\gamma$. Further,  let 
$$
\hat R(\alpha, \beta,\gamma)=\exp(-i\alpha J_z)\exp(-i\beta J_y)\exp(-i\gamma J_z) ,
$$
be the irreducible unitary action of the rotation group on the spin-$s$ space $\mathcal H^{(s)}$. Then the states 
$$
\hat R(\alpha, \beta,\gamma)|z, m\rangle
$$
are the eigenstates of 
$$
\hat R(\alpha, \beta,\gamma)e_z\cdot \vec J,
$$
where $e_z$ is the unit vector along the $z$-axis and $\vec J=(J_x, J_y, J_z)$~\cite{SA94}. 
We define
\begin{equation}\label{eq:ABbasis}
|a_m\rangle=\exp(-i\alpha m)|z,m\rangle,\quad |b_m\rangle=\exp(i\gamma m)\hat R(\alpha, \beta,\gamma)|z, m\rangle.
\end{equation}
Let us write $\EcalKDC(\alpha,\beta,\gamma)$ for the corresponding space of KD-positive states. Then the transition matrix $U(\alpha,\beta,\gamma)$ between these two bases has matrix elements
$$
U_{m'm}(\alpha,\beta,\gamma)=d^{(s)}_{m'm}(\beta),
$$
where $d^{(s)}(\beta)$ is Wigner's small $d$-matrix for spin-$s$~\cite{SA94}.  Note that the transition matrix depends only on the Euler angle $\beta$, not on the two others. Consequently, the same is true for $\EcalKDC(\alpha,\beta,\gamma)=\EcalKDC(\beta)$. 
Wigner's small $d$-matrix is real-valued so that the theory of the previous subsections applies. In particular, one never has $\EcalKDCext(\alpha,\beta,\gamma)=\Acal\cup\Bcal$ in this situation.

Let us now concentrate on the case $s=1$. One can then check that, with $\alpha_0=-\gamma_0=\frac{\pi}{4}$ and $\beta_0=\arccos(-1/3)$,
$$
R(\alpha_0,\beta_0,\gamma_0)=
\frac13\begin{pmatrix} -1&2&2\\2&-1&2\\ 2&2&-1
\end{pmatrix}.
$$
This unitary can be interpreted as a rotation by the angle $\pi$ about the axis $n=\frac{1}{\sqrt{3}}
\begin{pmatrix} 1\\1\\1\end{pmatrix}$.
The $ \{ |b_m\rangle \}$ therefore forms, for $m=-1,0,1$,  the eigenbasis of the observable $e_{z'}\cdot J$ where 
$$
e_{z'}=R(\alpha_0, \beta_0,\gamma_0) e_z,
$$
so that $e_{z'}\cdot J=\frac13(2J_x+2J_y-J_z)$.
Hence,
$$
U(\beta_0)=d^{(1)}(\beta_0)=\begin{pmatrix} \frac{1+\cos(\beta_0)}{2}&-\frac{\sin(\beta_0)}{\sqrt{2}}&\frac{1-\cos(\beta_0)}{2}\\ \frac{\sin(\beta_0)}{\sqrt{2}}&\cos(\beta_0)&-\frac{\sin(\beta_0)}{\sqrt{2}}\\ \frac{1-\cos(\beta_0)}{2}&\frac{\sin(\beta_0)}{\sqrt{2}}&\frac{1+\cos(\beta_0)}{2}
\end{pmatrix} = \frac{1}{3}\begin{pmatrix} 1&-2&2\\ 2&-1&-2\\ 2&2&1
\end{pmatrix}.
$$
Furthermore, 
$$
U(\beta_0)=d^{(1)}(\beta_0)=\begin{pmatrix}1&0&0\\0&-1&0\\0&0&-1\end{pmatrix}U^\star
\begin{pmatrix}-1&0&0\\0&-1&0\\0&0&1\end{pmatrix},
$$
where $U^\star$ is as in Eq.~\eqref{eq:specialU}.  
In conclusion, it then follows from Eq.~\eqref{eq:phases} that the matrix $U(\beta_0)$ is equivalent to the matrix $U^\star$.
As a result, if, for spin-$1$, the bases $\Acal$ and $\Bcal$ are as in Eq.~\eqref{eq:ABbasis}, with $\beta=\beta_0$, then the set of KD-positive states $\EcalKDC(\beta_0)$ is as described in the previous subsection. In particular, there then  exist mixed KD-positive states that are not mixtures of pure KD-positive states.

\subsection{Examples where $\convAB = \conv{\EcalKDCpu}\subsetneq \EcalKDC$}\label{s:mixedextremed=4n}
In this subsection, we show that there exist bases $\Acal$ and $\Bcal$ for which
\begin{equation}\label{eq:eureka3}
\Acal\cup\Bcal= \EcalKDCpu\subsetneq \EcalKDCext,
\end{equation}
so that 
\begin{equation}
\convAB = \conv{\EcalKDCpu}\subsetneq \EcalKDC.
\end{equation}
In  words, in this situation, the only pure KD-positive  states are the basis states, but those do not exhaust all extreme KD-positive  states: there   exist also mixed extreme KD-positive  states.  With the following proposition we show the occurrence of such situations in a large class of examples in dimensions $d=4n$ (with integer $n\leq166$) or $2^m$ with $m\in\N^*$ \cite{boykinetal2005,Ba22}. These examples are  less explicit than the three-dimensional example of the previous subsections.
\begin{Prop}
    Suppose that $U$ is a real-valued transition matrix for MUB bases in dimension $d\geqslant 4$. Then, there exists a real-valued and unitary matrix $W$ close to $U$, such that $W$ satisfies Eq. \eqref{eq:eureka3}.
\end{Prop}

\begin{proof}
We write $U=(\bracket{a_i^{U}}{b_j^{U}})_{(i,j)\in\IntEnt{1}{d}^2}$. We have that 
$$
\mab^{U} = \min_{(i,j)\in\IntEnt{1}{d}^2} \left|\bracket{a_i^{U}}{b_{j}^{U}}\right| = \frac{1}{\sqrt{d}}
$$ and 
$$
M_{\Acal,\Bcal}^{U} = \max_{(i,j)\in\IntEnt{1}{d}^2} \left|\bracket{a_i^{U}}{b_{j}^{U}}\right| = \frac{1}{\sqrt{d}}
$$ because $U$ is the transition matrix for MUB bases.

 It follows from Theorem~13 and (the proof of) Theorem~5 in~\cite{debievre2023a} that for all $\epsilon>0$ small enough, there exists a real unitary matrix $W=(\bracket{a_i^{W}}{b_j^{W}})_{(i,j)\in\IntEnt{1}{d}^2}$ such that $\norme{U-W}{\infty} \leqslant \epsilon$ and with the property that the only KD-positive pure states for $W$ are the basis states: $\EcalKDCpu=\mathcal A\cup \mathcal B$. 

Thus, $\mathrm{span}_{\R}(\EcalKDCpu) = \spanRAB$.  On the other hand, since $W$ is real, we know from Lemma~\ref{lem:realU} that $\VR= \SAOr$. Hence $ \mathrm{dim}(\VR)=\frac{d(d+1)}{2} > 2d-1 = \mathrm{dim}(\spanRAB)$, so that it follows from Proposition~\ref{prop:TSKD2bis} that $\convAB\subsetneq \EcalKDC$. Thus, Eq.~\eqref{eq:eureka3} holds.
\end{proof}

\section{Conclusions and discussion}\label{s:conc}

In recent years, the Kirkwood-Dirac quasiprobability distribution has risen as a versatile and powerful tool to study and develop protocols in discrete-variable quantum-information processing. Given two bases $\mathcal{A}$ and $\mathcal{B}$ and a state $\rho$, the associated KD distribution can be either a probability distribution or not. As reviewed in our Introduction, the existence of negative or nonreal entries in a KD distribution has been linked to quantum phenomena in several areas of quantum mechanics. This motivated us to investigate the divide between positive and nonpositive KD distributions. Previous studies~\cite{arvidsson-shukuretal2021, debievre2021, debievre2023a} have mapped out sufficient and necessary conditions for a pure state to assume a nonpositive KD distribution. But, to the best of our knowledge, no previous work has provided such an analysis for mixed states. 

In this work, we have presented the first thorough analysis of the set of mixed states that assume positive KD distributions. Our results can be grouped in two categories. 

\begin{itemize}
  \item Firstly, we have established that in several scenarios the set of KD positive states equals the convex combinations of the bases’ states $\mathcal{A}$ and $\mathcal{B}$. In particular, we have proven this to be the case for: $(i)$ any qubit ($d=2$) system provided  that $\mab >0$; $(ii)$ an open dense set of probability $1$ of possible choices of bases $\Acal$ and $\Bcal$ in dimension $3$; $(iii)$ prime dimensions, when the unitary transition matrix between the two bases $\Acal$ and $\Bcal$ is the discrete Fourier transform; and $(iv)$ any two bases that are sufficiently close to some other pair of bases for which the property holds. In addition to having shown that $\EcalKDC = \convAB$ for randomly chosen bases in dimension $d=3$, we conjecture that this is true also in higher dimensions $d\geq 4$.  We have given analytical and numerical evidence to that effect.   
  
  \item Secondly, we have proven that there exist scenarios where the set of KD-positive states includes mixed states that cannot be written as convex combinations of pure KD-positive states: in other words, we have shown that in such cases $\EcalKDC \neq \conv{\EcalKDCpu}$. This mirrors what happens for mixed Wigner positive states which are known not to all be mixtures of pure Wigner positive states~\cite{genonietal2013,vanherstraetencerf2021, hertzdebievre2023}. However, we go further by explicitly constructing, for a specific spin-$1$ system, extreme KD-positive states that cannot be written as convex mixtures of pure KD-positive states. For the Wigner distribution, such extreme positive states have not yet been constructed.   
\end{itemize}

Having a good understanding of the KD-positive states is a prerequisite for an efficient study of states for which the KD distribution takes negative or nonreal values. The latter are known to be related to nonclassical phenomena in various applications. To analyze the connection between mixed nonpositive states and nonclassicality, one should investigate measures and monotones of KD nonpositivity. In an upcoming paper, currently under construction, we use the findings of this work to analyze the so-called KD negativity~\cite{GoHaDr19}  of mixed states. Furthermore, our follow-up paper extends to mixed states the characterization of KD-positive states via their support uncertainty, as was done for pure states in~\cite{arvidsson-shukuretal2021, debievre2021, debievre2023a}.

\medskip
\noindent {\it Acknowledgments}: 
This work was supported in part by the Agence Nationale de la Recherche under grant ANR-11-LABX-0007-01 (Labex CEMPI), by the Nord-Pas de Calais Regional Council and the European Regional Development Fund through the Contrat de Projets \'Etat-R\'egion (CPER), and by  the CNRS through the MITI interdisciplinary programs. We thank Girton College, Cambridge, for support of this work. D.R.M. Arvidsson-Shukur thanks Nicole Yunger Halpern for useful discussions.

\appendix 
\appendixpage
\renewcommand{\theequation}{\thesection.\arabic{equation}}
\renewcommand{\thesection}{\Alph{section}}
\renewcommand{\thesubsection}{\Alph{section}.\arabic{subsection}}

\section{The geometric structure of $\convAB$}\label{s:convABstructure}
The polytope $\convAB$ has $2d$ vertices and lies in the $2(d-1)$-dimensional affine subspace of $\spanRAB$ determined by the constraint $\Tr F=1$. In this Appendix, we identify its interior and the facets making up its boundary. We will see the interior is not empty, so that the polytope is $2(d-1)$-dimensional. Its facets are therefore $(2d-3)$-dimensional. They are polytopes with $2(d-1)$ vertices. 
In particular, when $d=3$, $\convAB$ is a four-dimensional object. Its boundary has $9$ facets, which are three-dimensional tetrahedra.   

\begin{Lemma}
Let $\mab>0$. Let $\rho=\sum_i \lambda_i |a_i\rangle\langle a_i|+ \sum_j \mu_j |b_j\rangle\langle b_j|\in\convAB$, with $\lambda_i\geq0, \mu_j\geq 0$. Then 
\begin{equation}
\rho\in\mathrm{Int}\left(\convAB\right)\Leftrightarrow\lambda_1\lambda_2\dots\lambda_d\not=0\ \textrm{or}\ \mu_1\mu_2\dots\mu_d\not=0.
\end{equation}
\end{Lemma}
Note that the expression of $\rho$ as a convex combination of the basis vectors is not unique. What we are saying is that, if $\rho$ can be expressed in the manner stated, then it belongs to the interior of the polytope, and vice versa. 
\begin{proof} $\Rightarrow$ We will show the contrapositive. Suppose therefore that $\lambda_1\lambda_2\dots\lambda_d=0= \mu_1\mu_2\dots\mu_d$. We need to show $\rho\not\in \mathrm{Int}(\convAB)$. We can assume, without loss of generality, that $\lambda_d=0=\mu_d$.  Then $Q(\rho)_{dd}=0.$  Now consider, for $\epsilon>0$, $\delta\rho=\epsilon(|a_1\rangle\langle a_1|-|a_d\rangle\langle a_d|)$. 
Then, $\mathrm{Tr}(\rho+\delta\rho) = 1$, $\rho+\delta\rho\in\spanRAB$
and
$$
Q(\rho+\delta\rho)_{dd}=-\epsilon|\langle a_d|b_d\rangle|^2<0.
$$
Hence $\rho+\delta\rho\not\in\convAB$ so that $\rho$ belongs to the boundary of $\convAB$.\\
$\Leftarrow$
We consider the case where $\lambda_1\lambda_2\dots\lambda_d\not=0$, the other case being analogous. We can suppose without loss of generality that $\lambda_1=\min \lambda_i>0$. Then
$$
\sum_i\lambda_i|a_i\rangle\langle a_i|=
\frac{\lambda_1}{2}|a_1\rangle\langle a_1|+
\sum_{i\not=1}(\lambda_i-\frac{\lambda_1}{2})|a_i\rangle\langle a_i| 
+\frac{\lambda_1}{2}\sum_j |b_j\rangle\langle b_j|.
$$
Hence
$$
\rho= \frac{\lambda_1}{2}|a_1\rangle\langle a_1|+
\sum_{i\not=1}(\lambda_i-\frac{\lambda_1}{2})|a_i\rangle\langle a_i| 
+\sum_j(\mu_j+\frac{\lambda_1}{2}) |b_j\rangle\langle b_j|.
$$
Note that all coefficients are strictly positive.
Now consider a perturbation
$$
\delta\rho=\sum_i\epsilon_i|a_i\rangle\langle a_i|+\sum_j \delta_j |b_j\rangle\langle b_j|
$$
with $\epsilon_i, \delta_j\in\R$ and $\Tr\delta\rho=0$, then 
$$
\rho +\delta\rho \in \convAB
$$
provided the $\epsilon_i, \delta_j$ are small enough. 
In other words, there is a small ball centered on $\rho$ that belongs to $\convAB$.
\end{proof}

As a consequence of the proof, we also have the following result:
\begin{Cor}\label{cor:intconvAB}
Let $\mab>0$ and let $\rho$ be a density matrix. Then $\rho\in\mathrm{Int}(\convAB)$ if and only if there exist $\lambda_i>0, \mu_j>0$ so that
$\rho=\sum_i \lambda_i |a_i\rangle\langle a_i|+ \sum_j \mu_j |b_j\rangle\langle b_j|.$
\end{Cor}
Let us introduce the notation
$$
[i_1,i_2,\dots,i_k;j_1,j_2,\dots,j_{\ell}]=\conv{|a_{i_1}\rangle\langle a_{i_1}|,\dots,|a_{i_k}\rangle\langle a_{i_k}|, |b_{j_1}\rangle\langle b_{j_1}|, \dots, |b_{j_\ell}\rangle\langle b_{j_\ell}|},
$$
for any choice $1\leq i_1<i_2<\dots i_k\leq d$, $1\leq j_1<j_2<\dots <j_\ell\leq d$. Also, for $i\in \llbracket 1,d\rrbracket$, we write $\bar i =\llbracket 1,d\rrbracket\setminus \{i\}$. Then the Lemma implies that
$$
\partial(\convAB)=\cup_{i,j} [\bar i; \bar j].
$$
When $d=3$, this becomes
$$
\partial(\convAB)=\cup_{i_1<i_2;j_1<j_2} [i_1, i_2;j_1, j_2].
$$
The boundary is then the union of $9$ tetrahedra. They are glued together along 18 triangles of one of the following forms:  $[i_1, i_2; j]$ with $i_1 < i_2$ or $[i; j_1, j_2]$ with  $j_1 < j_2$.

\section{The geometry of $\EcalKDC$: the case of real orthogonal $U$ in $d=3$}
In this section, we give some more details about the geometry of the convex set $\EcalKDC$ of all KD-positive states in the particular case where $U$ is a real orthogonal matrix in dimension $d=3$ (Section~\ref{s:d3general}). 
We then analyze in detail the   example of Section~\ref{s:mixedextremed=3} for which both $\convAB \subsetneq \EcalKDC$ and $\mathrm{conv}(\EcalKDCpu) \subsetneq \EcalKDC$ (Section~\ref{subs:KDCextmixed}).  

\subsection{Identifying $\EcalKDC$.}\label{s:d3general}
We recall that, as in Section \ref{s:eurekanew}, in dimension 3, if $U$ is real, we can write
$$
\VR = \spanRAB\oplus \R F_{\perp}.
$$ 
Here, $F_{\perp}$ is orthogonal to the  $\spanRAB$.  We denote by $P_{\Acal,\Bcal}$ the orthogonal projection on  $\spanRAB$ associated to this decomposition, and by $T_{\perp} : F\in \VR \mapsto\mathrm{Tr}(F F_{\perp})$. Hence, the orthogonal projection of $F$ on $\R F_{\perp}$ is given by $T_{\perp}(F)F_{\perp}$. The following technical lemma and proposition collect the main properties of the set $\EcalKDC$ in this particular situation.

\begin{Lemma} \label{lem:EKDCorthV2}
Let $\sigma\in\spanRAB$. Then we have either $\left(\sigma + \R F_{\perp}\right)\cap \EcalKDC = \emptyset$ or there exists $-\infty < x_{-}(\sigma) \leqslant x_{+}(\sigma) < +\infty$ such that $$\left(\sigma + \R F_{\perp}\right) \cap \EcalKDC = \sigma + [x_{-}(\sigma),x_{+}(\sigma)] F_{\perp}.
$$
\end{Lemma}

\begin{proof}
If $\sigma \in\spanRAB$, then $\left(\sigma + \R F_{\perp}\right)\cap \EcalKDC$ is a compact convex set. Suppose the set is not empty. Therefore, as $T_{\perp}$ is continuous, $T_{\perp}(\left(\sigma + \R F_{\perp}\right)\cap \EcalKDC)$ is a non-empty compact interval of $\R$. This interval can be written as $\left(\sigma + \R F_{\perp}\right) \cap \EcalKDC = \sigma + [x_{-}(\sigma),x_{+}(\sigma)] F_{\perp}$ with $-\infty < x_{-}(\sigma) \leqslant x_{+}(\sigma) < +\infty$. 
\end{proof}

Let $\mcl{D} = P_{\Acal,\Bcal}(\EcalKDC)$ which is a subset of $\spanRAB$. Note that the second alternative of Lemma~\ref{lem:EKDCorthV2} happens if and only if $\sigma\in\mcl{D}$. In other words, $\mcl{D}$ is the domain of definition of $x_{-}$ and $x_{+}$. We will designate by $\mathrm{Int}(\mcl{D})$ the interior of $\mcl{D}$ as a subset of $\spanRAB$.

\begin{Prop}\label{prop:EKD2}
We have the following properties:
\begin{enumerate}[label=(\roman*),wide, labelindent=0pt]
	\item If $\sigma\in \mathrm{Int}(\convAB)$, $-\infty < x_{-}(\sigma) < 0 < x_{+}(\sigma) < +\infty$;
	\item If $\sigma\in\mcl{D}$ and $\sigma \notin \convAB$, then either $0<x_-(\sigma)\leq x_+(\sigma) < +\infty $ or $-\infty < x_-(\sigma)\leq x_+(\sigma)<0$;
	\item If $\sigma\in \Acal \cup \Bcal$ then $x_-(\sigma)=x_+(\sigma)=0$;
	\item The function $x_{+}$ has a maximum $x_{\max}$ on $\mcl{D}$. Moreover, the extreme points of 
 \[
 \mcl{Y}_{+}=\{\sigma + x_{\max}F_{\perp} \mid \sigma\in \mcl{D}, x_+(\sigma)=x_{\max}\}
 \]are extreme points of $\EcalKDC$;
	\item The function $x_{-}$ has a minimum $x_{\min}$ on $\mcl{D}$ . Moreover, the extreme points of 
 \[
 \mcl{Y}_{-}=\{\sigma + x_{\min}F_{\perp} \mid \sigma\in \mcl{D}, x_-(\sigma)=x_{\min}\}
 \]
 are extreme points of $\EcalKDC$;
	\item The function $x_{+}$ (resp. $x_{-}$) is concave (resp. convex) on $\mcl{D}$. Thus, it is continuous on $\mathrm{Int}(\mcl{D}).$ 
\end{enumerate}
\end{Prop}
In particular, the proposition implies  that $\EcalKDC$ lies between the ``bounding planes'':
$$
\{F\in \VR\mid \Tr F F_\perp=x_{\max}\}\quad\textrm{and}\quad 
\{F\in \VR\mid \Tr F F_\perp=x_{\min}\}.
$$
\begin{proof}
\begin{enumerate}[label=(\roman*),wide, labelindent=0pt]
 \item Suppose $\sigma\in \mathrm{Int}(\convAB)$, then, by Corollary \ref{cor:intconvAB}, $\sigma = \sum_{i=1}^{d} \lambda_{i}\cket{a_i}\bra{a_i} +  \mu_{i}\cket{b_i}\bra{b_i}$ with $\lambda_i > 0$ and $\mu_i >0$ for all $i\in\IntEnt{1}{d}$. Thus, $\min_{(i,j)\in\IntEnt{1}{d}^2} Q_{i,j}(\sigma) > 0$. Then, for $\epsilon > 0$,
\[
\forall (i,j)\in\IntEnt{1}{d}^2, \ Q_{i,j}(\sigma+\epsilon F_{\perp}) = Q_{i,j}(\sigma) + \epsilon Q_{i,j}(F_{\perp}) >  \min_{(i,j)\in\IntEnt{1}{d}^2} Q_{i,j}(\sigma) - \epsilon \max_{(i,j)\in\IntEnt{1}{d}^2} \left|Q_{i,j}(F_{\perp})\right|.
\]
Thus, there exists an $\epsilon_1 > 0$ such that for $\epsilon\in [0,\epsilon_1]$, $\sigma+\epsilon F_{\perp}\in \VRp$. Here, we recall that $\VRp$ is the set of self-adjoint operators with positive KD distributions. 

Moreover, for $\cket{\psi}\in\H_{1}$, for all $\epsilon \in \R^{+}$,
\[
\bra{\psi}\sigma + \epsilon F_{\perp}\cket{\psi} \geqslant \bra{\psi}\sigma\cket{\psi} - 
\epsilon \max_{\cket{\phi}\in\H_{1}}\left|\bra{\phi} F_{\perp}\cket{\phi}\right|\]
and
\[
\begin{array}{rcl}
\bra{\psi}\sigma\cket{\psi} = \sum_{i=1}^{d} \lambda_{i}\left|\bracket{\psi}{a_i}\right|^2 +  \mu_{i}\left|\bracket{\psi}{b_i}\right|^2 &\geqslant& \min_{i\in \IntEnt{1}{d}} \{\lambda_i,\mu_i\} \sum_{i=1}^{d} \left|\bracket{\psi}{a_i}\right|^2 +  \left|\bracket{\psi}{b_i}\right|^2 \\
&\geqslant& 2\min_{i\in \IntEnt{1}{d}} \{\lambda_i,\mu_i\} > 0.
\end{array}
\]
Consequently, there exists an $\epsilon_2 > 0$ such that for $\epsilon \in [0,\epsilon_2]$, for all $\cket{\psi}\in\H_{1}$, $\bra{\psi}\sigma + \epsilon F_{\perp}\cket{\psi} \geqslant 0$.
Therefore, for $\epsilon\in [0,\min(\epsilon_1,\epsilon_2)]$, $\sigma+\epsilon F_{\perp}$ is a density matrix with a positive KD distribution so that $0 < x_{+}(\sigma) $.

By changing $\epsilon$ to $-\epsilon$ in the previous lines, it follows that $x_{-}(\sigma)<0$. 
\item If $\sigma \notin \convAB$, then by Lemma~\ref{lem:G1}, $\sigma \notin \EcalKDC$ and thus, either $0 < x_{-}(\sigma)$ or $0 > x_{+}(\sigma)$.

\item Suppose $\sigma=\cket{a_1}\bra{a_1}$, then $\mathrm{Tr}((\sigma + xF_{\perp})^2) = 1+x^2$. For $x \neq 0$, $\mathrm{Tr}((\sigma + xF_{\perp})^2) > 1$ implying that $\sigma + xF_{\perp}$ is not a state. Hence, $\sigma + x F_{\perp}\notin \EcalKDC$ for all $x\neq 0$. Thus, $x_{+}(\sigma)=0=x_{-}(\sigma)$.

\item As $\EcalKDC$ is a compact set, $\mathrm{T}_{\perp}$ is bounded and reaches its bounds on $\EcalKDC$. Especially, it reaches its maximum $x_{\max}$ which is strictly positive. Note that  
\begin{equation}\label{eq:Ydef}
\mcl{Y}_{+} = \left\{\rho\in \SAOpone \mid \mathrm{T}_{\perp}(\rho)=\mathrm{Tr}(F_{\perp}\rho) = x_{\max}\right\}\cap \EcalKDC.
\end{equation}
Thus, $\mcl{Y}_{+}$ is compact, convex and not empty so it has an extreme point. Let $\rho_e$ be such an extreme point. We show, by contradiction, that $\rho_e$ is also an extreme point of $\EcalKDC$. Suppose that $\rho_e$ is not, and write $\rho_e = \lambda\rho_1 +(1-\lambda)\rho_2$ with $\rho_1,\rho_2\in\EcalKDC$ and $\lambda\in(0,1)$. So, $\mathrm{T}_{\perp}(\rho_1) \leqslant x_{\max}$ and $\mathrm{T}_{\perp}(\rho_2) \leqslant x_{\max}$. Now, suppose $\mathrm{T}_{\perp}(\rho_1) < x_{\max}$. Then, $
\mathrm{T}_{\perp}(\rho_e) < x_{\max} $,
which is a contradiction. Thus,  $\mathrm{T}_{\perp}(\rho_1)=\mathrm{T}_{\perp}(\rho_2)=x_{\max}$, which show that $\rho_1,\rho_2\in\mcl{Y}_{+}.$ As $\rho_e$ is an extreme point of $\mcl{Y}_{+}$, $\rho_1=\rho_2=\rho_e$ and so $\rho_e$ is an extreme point of $\EcalKDC$.

\item The proof is analogous to the one of (iv).

\item We show that $x_{+}$ is concave on its domain of definition. As $\mcl{D}$ is the projection of a compact convex set, it is a compact convex set. Take $\sigma,\sigma'\in \mcl{D}$ and $\lambda\in[0,1]$. We will show that $x_{+}(\lambda\sigma + (1-\lambda)\sigma') \geqslant \lambda x_{+}(\sigma) + (1-\lambda)x_{+}(\sigma') $. We have that
\[
\lambda \sigma + (1-\lambda) \sigma' + \left[\lambda x_{+}(\sigma) + (1-\lambda)x_{+}(\sigma')\right]F_{\perp} = \lambda ( \sigma + x_{+}(\sigma)F_{\perp}) +  (1-\lambda) ( \sigma' + x_{+}(\sigma')F_{\perp}).
\]
As a convex combination of KD-positive states, it is a KD-positive state such that $\lambda x_{+}(\sigma) + (1-\lambda)x_{+}(\sigma')\in[x_{-}(\lambda \sigma + (1-\lambda) \sigma'),x_{+}(\lambda \sigma + (1-\lambda) \sigma')]$. Therefore, $\lambda x_{+}(\sigma) + (1-\lambda)x_{+}(\sigma')\leqslant x_{+}(\lambda\sigma + (1-\lambda)\sigma')$. Thus, $x_{+}$ is a concave function on $\mcl{D}$. It is then continuous on $\mathrm{Int}(\mcl{D})$~\cite{hiriart-urrutylemarechal2001}. The same argument shows that $x_{-}$ is convex and thus also continuous on $\mathrm{Int}{\mcl{D}}$.

\end{enumerate}
\end{proof}

\subsection{Identifying $\EcalKDCext$ : an example of mixed extreme states.}\label{subs:KDCextmixed}

We now identify some extreme mixed states of $\EcalKDC$ for the unitary matrix 
\begin{equation}\label{eq:Uspecial}
U=\frac{1}{3}\begin{pmatrix}
-1 & 2 & 2 \\
2 & -1 & 2 \\
2 & 2 & -1
\end{pmatrix},
\end{equation}
introduced in Section~\ref{s:mixedextremed=3},  using Proposition~\ref{prop:EKD2}~(iv).
Note that we identified all pure KD-positive states for $U$ in Section~\ref{s:mixedextremed=3} and that they all lie below $\spanRAB$: if $\rho\in \EcalKDCpu$, then $\Tr(\rho F_\perp)\leq0$.  So any $\rho\in\EcalKDCext$ for which $\Tr(\rho F_\perp)>0$ is a mixed extreme KD-positive state. We explicitly find some of those states as follows. 
 We first determine in Lemma~\ref{Lem:EKDCmax} the maximum $x_{\max}$ of the function $x_{+}$, which is strictly positive. This allows us to give a precise description of  the set $\mathcal Y_+$ in Proposition~\ref{prop:EKD2}, and in particular of its extreme points.


\begin{Lemma}\label{Lem:EKDCmax}
Let $U$ be as in Eq.~\eqref{eq:Uspecial}. Then, the maximum $x_{\max}$ of $x_+$ on $\mcl{D}$ is 
\[
x_{\max}=\max_{\sigma\in\mathcal D} x_+(\sigma) = \frac{1}{2\sqrt{6}}.
\]
This value is reached for $\sigma=\frac{1}{3}\bbone_{3}\in\convAB$.
\end{Lemma}
\begin{proof}
We set $\sigma\in \spanRAB$ with $\Tr(\sigma)=1$ so that
\[
\sigma = \sum_{i=1}^{3} \lambda_{i}\cket{a_i}\bra{a_i} + \mu_{i}\cket{b_i}\bra{b_i} \ \mathrm{and} \ \sum_{i=1}^{3}  \lambda_i + \mu_i = 1. 
\]
For all $x\in\R$, we  compute
\[
Q(\sigma+x F_{\perp}) = \frac{1}{27}
\begin{pmatrix}
3(\mu_1+\lambda_1)-2x\sqrt{6} & 12(\mu_2+\lambda_1)+x\sqrt{6} & 12(\mu_3+\lambda_1)+x\sqrt{6} \\
12(\mu_1+\lambda_2)+x\sqrt{6} & 3(\mu_2+\lambda_2)-2x\sqrt{6}& 12(\mu_3+\lambda_2)+x\sqrt{6}\\
12(\mu_1+\lambda_3)+x\sqrt{6} & 12(\mu_2+\lambda_3)+x\sqrt{6} & 3(\mu_3+\lambda_3)-2x\sqrt{6}
\end{pmatrix}.
\]
Thus, if 
\[
3(\mu_1+\lambda_1)-2x\sqrt{6} < 0 \ \mathrm{or} \ 3(\mu_2+\lambda_2)-2x\sqrt{6} < 0 \ \mathrm{or} \ 3(\mu_3+\lambda_3)-2x\sqrt{6} < 0, 
\]
or equivalently, if 
\[
x > \frac{3}{2\sqrt{6}}\min_{i\in\IntEnt{1}{3}} (\mu_i+\lambda_i),
\]
then $\sigma+x F_{\perp}$ is not KD-positive. Hence, since $\sigma+x_{+}(\sigma) F_{\perp}$ is KD-positive,  
\begin{equation}\label{eq:xpsigma}
x_{+}(\sigma) \leqslant \frac{3}{2\sqrt{6}}\min_{i\in\IntEnt{1}{3}} (\mu_i+\lambda_i).
\end{equation}
Thus, 
\[
x_{+}(\sigma) \leqslant \frac{3}{2\sqrt{6}}\inf_{\{\lambda_i,\mu_i\}}\min_{i\in\IntEnt{1}{3}} (\mu_i+\lambda_i),
\]
where the infimum is taken over all the $(\lambda_i,\mu_i)_{i\in\IntEnt{1}{3}}$ such that 
$\sigma = \sum_{i=1}^{3} \lambda_i\cket{a_i}\bra{a_i} + \mu_i\cket{b_i}\bra{b_i}.$
Moreover, $\min_{i\in\IntEnt{1}{3}} (\mu_i+\lambda_i) \leqslant \frac{1}{3}$ for any $\sigma = \sum_{i=1}^{3} \lambda_i\cket{a_i}\bra{a_i} + \mu_i\cket{b_i}\bra{b_i} \in \spanRAB$ with $\Tr(\sigma)=1$. Indeed, suppose there exists a $\sigma$ such that $\min_{i\in\IntEnt{1}{3}} (\mu_i+\lambda_i) > \frac{1}{3}$, then $\mathrm{Tr}(\sigma) \geqslant 3 \min_{i\in\IntEnt{1}{3}} (\mu_i+\lambda_i) > 1$, which is a contradiction. Thus,
\[
\inf_{\{\lambda_i,\mu_i\}}\min_{i\in\IntEnt{1}{3}} (\mu_i+\lambda_i)
\leqslant \frac{1}{3}.
\]
Therefore,
\[
x_{+}(\sigma) \leqslant \frac{1}{2\sqrt{6}}.
\]
Consequently, as the bound does not depend on $\sigma$, we obtain that
\[
x_{\max} \leqslant\frac{1}{2\sqrt{6}}.
\]
It is proven, in Section~\ref{s:mixedextremed=3}, that $x_{+}(\frac{1}{3}\bbone_3)= \frac{1}{2\sqrt{6}}$. Consequently, $x_{\max}= \frac{1}{2\sqrt{6}}$.

\end{proof}

\begin{Prop}\label{prop:Yplus}
The set $\mathcal Y_+=\EcalKDC \cap \{\rho\in \SAOpone \mid \mathrm{Tr}(\rho F_{\perp}) = \frac{1}{2\sqrt{6}}\}$ is of the form 
\[
\left\{\frac{1}{3}\bbone_{d} +\frac{1}{2\sqrt{6}}F_{\perp} + \lambda_1(\cket{a_1}\bra{a_1} - \cket{b_1}\bra{b_1}) + \lambda_2(\cket{a_2}\bra{a_2} - \cket{b_2}\bra{b_2}) \right\}, 
\]
with $\left|\lambda_2 - \lambda_1 \right| \leqslant \frac{3}{8}$, $\left|\lambda_1\right| \leqslant \frac{3}{8}$ and $\left|\lambda_2\right| \leqslant \frac{3}{8}$.
Its extreme points are obtained for the following  values of $(\lambda_1, \lambda_2)$:
\begin{equation}\label{eq:extremepoints}
 \left\{\left(0,\frac{3}{8}\right),\left(0,-\frac{3}{8}\right), \left(\frac{3}{8},0\right), \left(-\frac{3}{8},0\right), \left(\frac{3}{8},\frac{3}{8}\right), \left(-\frac{3}{8},-\frac{3}{8}\right)\right\}.
\end{equation}
\end{Prop}
\begin{proof}
Let $\rho\in\mcl{Y}_{+}$. Then there exists $\sigma\in\mathcal D$ so that $\rho=\sigma+\frac{1}{2\sqrt6}F_\perp$. In addition, $x_+(\sigma)=\frac1{2\sqrt6}$. Then,  Eq.~\eqref{eq:xpsigma} implies that there exist $\lambda_i,\mu_i\in\R$ so that $$\sigma=\sum_i(\lambda_i|a_i\rangle\langle a_i|+ \mu_i|b_i\rangle\langle b_i|)
$$
and so that 
$\min_{i\in\IntEnt{1}{3}} \mu_i+\lambda_i = \frac{1}{3}$. Indeed, if such a decomposition does not exist, as $\sigma\in \spanRAB$, all decompositions 
\[
\sigma=\sum_i(\alpha_i|a_i\rangle\langle a_i|+ \beta_i|b_i\rangle\langle b_i|)
\] 
satisfy $\min_{i\in\IntEnt{1}{3}} (\alpha_i+\beta_i) < \frac{1}{3}$. Thus, as shown in Eq~\eqref{eq:xpsigma},
\[
x_{+}(\sigma) \leqslant \frac{3}{2\sqrt{6}}\min_{i\in\IntEnt{1}{3}} (\alpha_i+\beta_i) < \frac{1}{2\sqrt{6}},
\]
which is a contradiction.

 Since $\sum_i(\lambda_i+\mu_i)=1$, this implies $\lambda_i+\mu_i=\frac13$ for $i=1,2,3$. It follows that there exist $(\lambda_i)_{i\in\IntEnt{1}{3}}$ so that
\[
\sigma= \frac{1}{3}\bbone_{3} +\lambda_1(\cket{a_1}\bra{a_1} - \cket{b_1}\bra{b_1}) + \lambda_2(\cket{a_2}\bra{a_2} - \cket{b_2}\bra{b_2}) + \lambda_3(\cket{a_3}\bra{a_3} - \cket{b_3}\bra{b_3}).
\]

As $\sum_{i=1}^{3} \cket{a_i}\bra{a_i}= \sum_{i=1}^{3} \cket{b_i}\bra{b_i}$, we can simplify the expression to obtain
\begin{equation}\label{eq:sigmaform}
\sigma= \frac{1}{3}\bbone_{3} + \lambda_1(\cket{a_1}\bra{a_1} - \cket{b_1}\bra{b_1}) + \lambda_2(\cket{a_2}\bra{a_2} - \cket{b_2}\bra{b_2}).
\end{equation}
The KD distribution of $\rho=\sigma+\frac{1}{2\sqrt{6}}F_{\perp}$ is
\[
Q(\sigma+\frac{1}{2\sqrt{6}}F_{\perp}) = 
\frac{1}{27}\begin{pmatrix}
0& 4 +12(\lambda_1-\lambda_2)+\frac{1}{2} & 4 +12\lambda_1+\frac{1}{2} \\
 4 - 12(\lambda_1-\lambda_2)+\frac{1}{2} & 0& 4+12\lambda_2+\frac{1}{2}\\
4 - 12\lambda_1+\frac{1}{2} & 4-12\lambda_2+\frac{1}{2} & 0
\end{pmatrix}.
\]

Since $\rho\in\mathcal Y_{+}$, it is KD positive, which is equivalent to  
\[
-4+12|\lambda_1|\leqslant \frac{1}{2}, \ -4+12|\lambda_2| \leqslant \frac{1}{2}, \ \mathrm{ and } \ -4+12|\lambda_1 -\lambda_2| \leqslant \frac{1}{2},
\]
or
\begin{equation}\label{eq:condition1}
|\lambda_1| \leqslant \frac{3}{8}, |\lambda_2| \leqslant \frac{3}{8}, \ \mathrm{ and } \ |\lambda_1 -\lambda_2|\leqslant \frac{3}{8}.
\end{equation}
The eigenvalues of $\rho=\sigma+\frac{1}{2\sqrt{6}}F_{\perp}$ are :
\[
\left\{\frac{1}{4}, \frac{3}{8}+\frac{1}{24}\sqrt{9+8^3(\lambda_1^2+\lambda_2^2-\lambda_1\lambda_2)}, \frac{3}{8}-\frac{1}{24}\sqrt{9+8^3(\lambda_1^2+\lambda_2^2-\lambda_1\lambda_2})\right\}.
\]
Consequently, $\rho=\sigma+\frac{1}{2\sqrt{6}}F_{\perp}$ is a positive operator if and only if 
\[
\frac{3}{8}-\frac{1}{24}\sqrt{9+8^3(\lambda_1^2+\lambda_2^2-\lambda_1\lambda_2)} \geqslant 0.
\]
This equation is equivalent to
\begin{equation}\label{eq:condition2}
\lambda_1^2+\lambda_2^2+(\lambda_1-\lambda_2)^2 \leqslant 2\left(\frac{3}{8}\right)^2.
\end{equation}
We have therefore established  that $\rho\in\mcl{Y}_{+}$ if and only if it can be written as $\rho=\sigma+\frac{1}{2\sqrt6}F_\perp$ with $\sigma$ as in Eq.~\eqref{eq:sigmaform} and with $\lambda_1, \lambda_2$ satisfying Eq.~\eqref{eq:condition1}-\eqref{eq:condition2}. 
The set $\mcl{C}' = \{(\lambda_1,\lambda_2)\in\R^2 \mid \lambda_1^2+\lambda_2^2+(\lambda_1-\lambda_2)^2 \leqslant 2\left(\frac{3}{8}\right)^2 \}$ is bounded by an ellipse and is as such convex.  The set $\mcl{C}=\{(\lambda_1,\lambda_2)\in\R^2 \mid |\lambda_1| \leqslant \frac{3}{8}, |\lambda_2| \leqslant \frac{3}{8} \ \mathrm{ and } \ |\lambda_1 -\lambda_2|\leqslant \frac{3}{8} \}$ is a convex hexagon. Its extreme points are identified to be those given in~Eq.~\eqref{eq:extremepoints}.

Noting that these extreme points lie on the ellipse bounding $\mathcal C'$, we conclude that in fact $\mcl{C}\subsetneq\mcl{C'}$; see Fig.~\ref{fig:extremepoints}. This proves our Proposition.

\end{proof}

\begin{figure}
    \centering
    \includegraphics[scale=0.5]{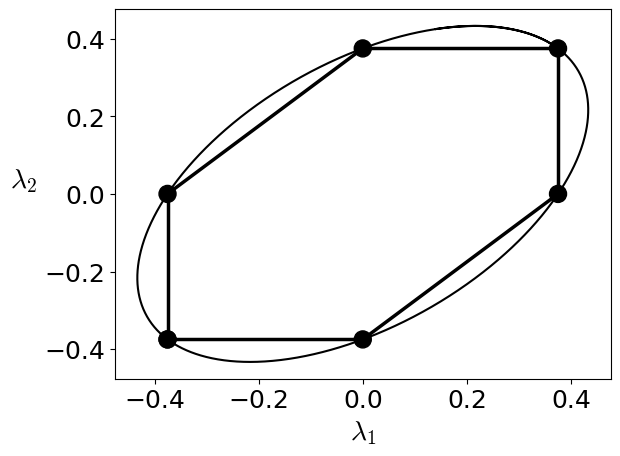}
    \caption{The convex regions $\mcl{C}$ and $\mcl{C}'$ in the $(\lambda_1, \lambda_2)$-plane defined by Eq.~\eqref{eq:condition1}-\eqref{eq:condition2}. Dots represent the extreme points of $\mcl{Y}_{+}$}
    \label{fig:extremepoints}
\end{figure}

\section{Pure KD-positive states for MUBs}\label{s:MUBpure}
The pure KD-positive states of MUB bases can be characterized as follows.

\begin{Theorem}\label{Thm:MUB1}
Suppose  $\Acal$ and $\Bcal$ are MUB bases. Then:\\
(i)  A pure state $\cket{\psi}$ is KD positive iff $\na(\psi)\nb(\psi)=d$;\\
(ii) If $d$ is a prime number, then the only pure KD-positive states are the basis states.\\
\end{Theorem}
This result is implicit in
~\cite{Xu22}. We provide a simple proof below. Note that this result implies that, when $d$ is a prime number, then the only pure KD-positive states of MUB bases are their basis states. This last result was proven for the DFT in~\cite{debievre2023a}, where the same result is also obtained for perturbations of MUB bases that are completely incompatible, a notion introduced in~\cite{debievre2021}. It is not known, to the best of our knowledge, if under the hypotheses of the theorem, there do also exist mixed KD-positive states.  
\begin{proof}
Suppose $\cket{\psi}$ is a KD-positive state. By permuting the order and changing the phases of basis states, we can suppose that 
\[
S_{\Acal}(\psi)= \IntEnt{1}{\na(\psi)},S_{\Bcal}(\psi)= \IntEnt{1}{\nb(\psi)}, \left(\bracket{a_i}{\psi}\right)_{i\in\IntEnt{1}{d}}\in\left(\R^{+}\right)^{d} \text{ and }  \left(\bracket{b_j}{\psi}\right)_{j\in\IntEnt{1}{d}}\in\left(\R^{+}\right)^{d}.
\]
Here $S_{\Acal}(\psi)=\{i\in\IntEnt{1}{d}, \bracket{a_i}{\psi}\neq 0\}$ and $\na(\psi)=\sharp S_{\Acal}(\psi)$. The same definitions hold for $\Bcal$. 
Hence, since $U$ is the transition matrix for MUB bases and since  $Q(\psi)\in \left(\R^{+}\right)^{d^2}$, one concludes that $\forall (i,j)\in\IntEnt{1}{\na(\psi)}\times\IntEnt{1}{\nb(\psi)}, \bracket{a_i}{b_j}=\frac{1}{\sqrt{d}}$. 
By construction, one has
\[
\forall i\in\IntEnt{1}{\na(\psi)}, \bracket{a_i}{\psi} = \sum_{j=1}^{\nb(\psi)} \bracket{a_i}{b_j}\bracket{b_j}{\psi} =  \frac{1}{\sqrt{d}}\sum_{i=1}^{\nb(\psi)}\bracket{b_j}{\psi}
\]
which is independent of $i$. Similarly,
\[
\forall j\in\IntEnt{1}{\nb(\psi)}, \bracket{b_j}{\psi} = \sum_{i=1}^{\na(\psi)} \bracket{b_j}{a_i}\bracket{a_i}{\psi} = \frac{1}{\sqrt{d}}\sum_{i=1}^{\na(\psi)}\bracket{a_i}{\psi}
\]
which is independent of $j$ so that
\[
\bracket{b_1}{\psi} =  \frac{1}{\sqrt{d}}\sum_{i=1}^{\na(\psi)}\bracket{a_i}{\psi} =  \frac{1}{d}\sum_{i=1}^{\na(\psi)}\sum_{i=1}^{\nb(\psi)}\bracket{b_j}{\psi} =\frac{\bracket{b_1}{\psi}}{d}\sum_{i=1}^{\na(\psi)}\sum_{i=1}^{\nb(\psi)}1 = \frac{\bracket{b_1}{\psi}}{d}\na(\psi)\nb(\psi).
\]
As $\bracket{b_1}{\psi}\neq 0$, one finds that
\[
\na(\psi)\nb(\psi) = d.
\]
In particular, when $d$ is prime, this implies that $\na(\psi) =1$ or $\nb(\psi) =1$. In either case, $\cket{\psi}$ is a basis state.
\end{proof}

\bibliographystyle{unsrt}
\bibliography{BIB_Article_Geometry_KD_states_2023_05_31}

\end{document}